\documentclass[]{llncs}
\usepackage[utf8]{inputenc}
\usepackage{amssymb}
\usepackage{amsmath}
\usepackage{xspace}
\usepackage{graphicx}
\usepackage{subfigure}
\usepackage{complexity}
\usepackage{thm-restate}
\usepackage{hyperref}
\usepackage[capitalise]{cleveref}
\usepackage{todonotes}
\usepackage{enumerate}
\usepackage{paralist}
\usepackage[framemethod=tikz]{mdframed}
\usepackage{bibentry}

\graphicspath{{figs/}}

\crefname{property}{Property}{Properties}
\crefname{claim}{Claim}{Claims}
\newenvironment{nestedproof}{\noindent{\emph{Proof:}}}{~\hfill $\blacksquare$ \medskip}

\newcommand{\myparagraph}[1]{\medskip\noindent\textbf{\boldmath #1}}
\newmdenv[%
  roundcorner=5pt,%
  backgroundcolor=orange,%
]{todonote}

\DeclareMathOperator{\ch}{ch}
\DeclareMathOperator{\cross}{cr}

\let\doendproof\endproof
\renewcommand\endproof{~\hfill$\qed$\doendproof}

\begin{document}

\title{Min-$k$-planar Drawings of Graphs\thanks{Research started at the Summer Workshop on Graph Drawing (SWGD) 2022, and partially supported by: (i)University of Perugia, Ricerca Base 2021, Proj. AIDMIX and RICBA22CB ; (ii) MUR PRIN Proj. 2022TS4Y3N - ``EXPAND: scalable algorithms for EXPloratory Analyses of heterogeneous and dynamic Networked Data''; (iii) MUR PRIN Proj. 2022ME9Z78 - ``NextGRAAL: Next-generation algorithms for constrained GRAph visuALization''
}
}

\author{Carla~Binucci\inst{1}\texorpdfstring{\href{https://orcid.org/0000-0002-5320-9110}{\protect\includegraphics[scale=0.45]{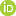}}}{}
\and Aaron~B\"ungener\inst{2} 
\and Giuseppe~Di~Battista\inst{3}\texorpdfstring{\href{https://orcid.org/0000-0003-4224-1550}{\protect\includegraphics[scale=0.45]{orcid}}}{} \and Walter~Didimo\inst{1}\texorpdfstring{\href{https://orcid.org/0000-0002-4379-6059}{\protect\includegraphics[scale=0.45]{orcid}}}{} \and Vida~Dujmovi\'c\inst{4}\texorpdfstring{\href{https://orcid.org/0000-0001-7250-0600}{\protect\includegraphics[scale=0.45]{orcid}}}{} \and Seok-Hee~Hong\inst{5}\texorpdfstring{\href{https://orcid.org/0000-0003-1698-3868}{\protect\includegraphics[scale=0.45]{orcid}}}{} \and Michael~Kaufmann\inst{2}\texorpdfstring{\href{https://orcid.org/0000-0001-9186-3538}{\protect\includegraphics[scale=0.45]{orcid}}}{} \and Giuseppe~Liotta\inst{1}\texorpdfstring{\href{https://orcid.org/0000-0002-2886-9694}{\protect\includegraphics[scale=0.45]{orcid}}}{} \and Pat~Morin\inst{6}\texorpdfstring{\href{https://orcid.org/0000-0003-0471-4118}{\protect\includegraphics[scale=0.45]{orcid}}}{} \and Alessandra~Tappini\inst{1}\texorpdfstring{\href{https://orcid.org/0000-0001-9192-2067}{\protect\includegraphics[scale=0.45]{orcid}}}{}}

\date{}

\institute{Universit\`a degli Studi di Perugia, Perugia, Italy\\
  \texttt{\{carla.binucci,walter.didimo,giuseppe.liotta,alessandra.tappini\}@unipg.it}
  \and
  University of T\"ubingen, T\"ubingen, Germany\\
  \texttt{michael.kaufmann@uni-tuebingen.de},
  \texttt{aaron.buengener@student.uni-tuebingen.de},
  \and
  Universit\`a Roma Tre, Roma, Italy\\
  \texttt{giuseppe.dibattista@uniroma3.it}
  \and
  University of Ottawa, Ottawa, Canada\\  
  \texttt{vdujmovi@uOttawa.ca}
  \and
  University of Sydney, Sydney, Australia\\
  \texttt{seokhee.hong@sydney.edu.au}
  % to avoid spam, I don't put my email address in articles
  \and
  Carleton Univerity, Ottawa, Canada\\
  \texttt{morin@scs.carleton.ca}
}

\authorrunning{C. Binucci et al.}

\maketitle

\begin{abstract}
 The study of nonplanar drawings of graphs with restricted crossing configurations is a well-established topic in graph drawing, often referred to as \emph{beyond-planar graph drawing}. One of the most studied types of drawings in this area are the \emph{$k$-planar drawings} $(k \geq 1)$, where each edge cannot cross more than $k$ times. We generalize $k$-planar drawings, by introducing the new family of \emph{min-$k$-planar drawings}. In a min-$k$-planar drawing edges can cross an arbitrary number of times, but for any two crossing edges, one of the two must have no more than $k$ crossings. We prove a general upper bound on the number of edges of min-$k$-planar drawings, a finer upper bound for $k=3$, and tight upper bounds for $k=1,2$. Also, we study the inclusion relations between min-$k$-planar graphs (i.e., graphs admitting min-$k$-planar drawings) and $k$-planar graphs. In our setting we only allow \emph{simple} drawings, that is, any two edges cross at most once, no two adjacent edges cross, and no three edges intersect at a common crossing point.
 %, and other beyond-planar graph classes.
 
 %\and beyond-planarity inclusion relationships}
\end{abstract} 

%%%%%%%%%%%%%%%%%%%%%%%%%%%%%%%%%%%%%%%
%%%%%%%%%%%% INTRODUCTION %%%%%%%%%%%%%
%%%%%%%%%%%%%%%%%%%%%%%%%%%%%%%%%%%%%%%
\section{Introduction}\label{se:intro}
Beyond planarity~\cite{DBLP:journals/csur/DidimoLM19,DBLP:books/sp/20/HT2020} is a recent area of focus in graph drawing and topological graph theory, having its foundations established in the 1970s and 1980s.
It comprises works on graphs that go beyond planar graphs in the sense that several, mostly local, crossing configurations are forbidden. 
% when being drawn in the plane. 
The simplest are \emph{1-planar} graphs, where at most one crossing per edge is allowed~\cite{DBLP:journals/csr/KobourovLM17,MR0187232}, and their generalization \emph{$k$-planar} graphs, where at most $k \geq 1$ crossings per edge are tolerated~\cite{DBLP:conf/compgeom/Bekos0R17,DBLP:journals/csur/DidimoLM19,DBLP:journals/algorithmica/GrigorievB07,DBLP:conf/gd/PachT96,DBLP:journals/combinatorica/PachT97}.
Other prominent examples of graph classes are \emph{fan-planar} graphs~\cite{DBLP:conf/gd/BekosCGHK14,DBLP:journals/algorithmica/BekosCGHK17,DBLP:journals/tcs/BinucciGDMPST15,DBLP:conf/gd/BinucciGDMPT14,DBLP:journals/combinatorics/0001U22}, where several edges might cross the same edge but they should be adjacent to the same vertex, and \emph{$k$-gap-planar} graphs ($k \geq 1)$~\cite{DBLP:conf/gd/BachmaierRS18,DBLP:conf/gd/BaeBCEEGHKMRT17,DBLP:journals/tcs/BaeBCEE0HKMRT18}, where for each pair of crossing edges one of the two edges contains a small gap through which the other edge can pass, and at most $k$ gaps per edge are allowed. Another popular family is the one of \emph{$k$-quasiplanar} graphs, which forbids $k$ mutually crossing edges~\cite{ackerman.tardos:on,DBLP:conf/gd/AgarwalAPPS95,DBLP:journals/combinatorica/AgarwalAPPS97,DBLP:journals/jctb/AngeliniBBLBDHL20,DBLP:journals/siamdm/FoxPS13}. Mostly, edge density and inclusion relations of different beyond-planar graph classes have been studied~\cite{DBLP:journals/jctb/AngeliniBBLBDHL20,DBLP:journals/csur/DidimoLM19,DBLP:books/sp/20/HT2020}.

\smallskip
In this paper we introduce a new graph family that generalizes $k$-planar graphs by permitting certain edges to have more than $k$ crossings. Namely, for each two crossing edges we require that at least one of them contains at most $k$ crossings. Formally, this graph family is defined as follows:

\begin{definition}\label{def:mink}
A graph $G$ is \emph{min-$k$-planar} $(k \geq 1)$ if it admits a drawing on the plane, called \emph{min-$k$-planar} drawing, such that for any two crossing edges $e$ and $e'$ it holds $\min\{\cross(e), \cross(e')\} \leq k$ , where $\cross(e)$ and $\cross(e')$ are the number of crossings of $e$ and $e'$, respectively.
\end{definition}

Clearly, every $k$-planar drawing $\Gamma$ is also min-$k$-planar, but not vice versa. A \emph{crossing edge} in $\Gamma$ with more than $k$ crossings is \emph{heavy}, otherwise it is \emph{light}. There are two main motivations behind the study of min-$k$-planar graphs:

%\begin{itemize}
\smallskip \noindent $(i)$ From a theoretical perspective, when a graph is not $k$-planar we may want to draw it by allowing some heavy edges, whose removal yields a $k$-planar drawing. In this respect, if $m$ is the total number of edges in the graph, we will prove that the number of heavy edges in a min-$k$-planar drawing is at most $\frac{k}{2k+1}\cdot m$, which varies in the interval $[\frac{m}{3}, \frac{m}{2})$. 
% We remark that a min-k planar graph is a graph that becomes k-planar upon removal of some ``heavily crossed edges'' which, conceptually speaking can be related with the notion of graph skewness 
% for planar graphs \cite{DBLP:journals/algorithmica/CabelloM11,DBLP:journals/siamcomp/CabelloM13}.

%Hence,  min-$k$-planar graphs can be regarded as ``near $k$-planar'' graphs, thus extending the well-known concept of ``near planar'' graph (graph with bounded skewness) \cite{DBLP:journals/algorithmica/CabelloM11,DBLP:journals/siamcomp/CabelloM13}.

\smallskip\noindent $(ii)$  From a practical perspective, even if a graph is $k$-planar, allowing (few) pairwise-independent heavy edges may reduce the visual complexity of the layout, even when the total number of crossings grows. For example, \cref{fi:intro} shows two drawings of the same portion of a graph. Despite the drawing in \cref{fi:intro-a} is 2-planar and has fewer crossings in total, the one in \cref{fi:intro-b} appears more readable; it is not 2-planar, but it is min-2-planar. 
%\end{itemize}

\smallskip Min-$k$-planar graphs are also implicitly studied in \cite{DBLP:conf/gd/WoodT06,wt-07}, proving that the underlying graph of a convex min-$k$-planar drawing has treewidth~$3k+11$. 

\begin{figure}[tb]
	\centering
    \subfigure[2-planar drawing]{\includegraphics[page=1,width=.35\textwidth]{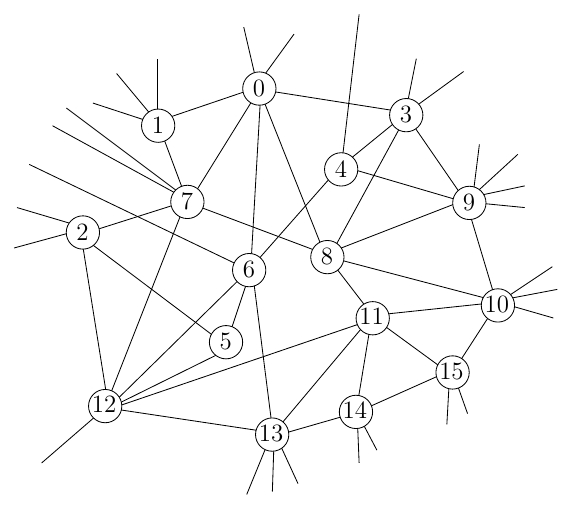}\label{fi:intro-a}}
    \hfil
    \subfigure[min-2-planar drawing]{\includegraphics[page=2,width=.35\textwidth]{intro}\label{fi:intro-b}}
    \caption{Two drawings of the same portion of a graph: (a) is 2-planar and has 10 crossings; (b) is min-2-planar, is not 2-planar, and has 12 crossings; it contains two ``heavy'' edges incident to vertex 6, each with several crossings.}
    \label{fi:intro}
\end{figure}

\myparagraph{Contribution.} We study the edge density of min-$k$-planar graphs (\cref{se:edge-density}) and their inclusion relations with $k$-planar graphs (\cref{se:inclusion-relationships}). In our setting we only allow \emph{simple} drawings, that is, any two edges cross at most once, no two adjacent edges cross, and no three edges intersect at a common crossing point. After giving general bounds on edge and crossing numbers, we focus on $k \in \{1,2,3\}$:

%\begin{itemize}
%\item 
\smallskip \noindent {$-$} We provide tight upper bounds on the maximum number of edges of min-1-planar and min-2-planar graphs. Namely, we prove that $n$-vertex min-1-planar graphs and min-2-planar graphs have at most $4n-8$ edges and at most $5n-10$ edges, respectively, as for 1-planar and 2-planar graphs. For min-3-planar graphs we give an upper bound of $6n-12$ and show min-3-planar graphs with $5.6n-O(1)$ edges, hence having density higher than the one of every 3-planar graph. 

%\item
\smallskip \noindent {$-$}
Despite the maximum density of min-$k$-planar graphs for $k=1,2$ equals the one of $k$-planar graphs, we show that $1$-planar and $2$-planar graphs are proper sub-classes of min-1-planar and min-2-planar graphs (as for $k=3$). However, the min-1-planar graphs that can reach the maximum density of $4n-8$ are also 1-planar (i.e., the two classes coincide), while this is not true for $k = 2, 3$. 
 
% \smallskip \noindent {$(iii)$}
% \textcolor{red}{Revise this paragraph, depending on the results.} As a further contribution, we study the relationships between min-$k$-planar graphs and other beyond-planar graph classes. In particular, we show that min-$k$-planar graphs and fan-planar graphs are incomparable classes, for any $k\geq2$. We remark that, as for min-2-planar graphs, also fan-planar graphs have at most $5n-10$ edges \cite{DBLP:journals/combinatorics/0001U22}. 

\smallskip
\Cref{se:basic-definitions} introduces notation and terminology; final remarks and open problems are in \Cref{se:open}.
All full proofs are in the appendix.

%%%%%%%%%%%%%%%%%%%%%%%%%%%%%%%%%%%%%%%
%%%%%%% BASIC DEFINITIONS %%%%%%%%%%%%%
%%%%%%%%%%%%%%%%%%%%%%%%%%%%%%%%%%%%%%%
\section{Basic Definitions}\label{se:basic-definitions}

We only deal with connected graphs. A graph is \emph{simple} if it does not contain multiple edges and self-loops. A graph with multiple edges but not self-loops is also called a \emph{multi-graph}.
%We assume that $G$ is connected and simple, meaning that it contains neither multiple edges nor self-loops. 
Let $G$ be any (not necessarily simple) graph. We denote by $V(G)$ and $E(G)$ the set of vertices and the set of edges of~$G$, respectively. A \emph{drawing} $\Gamma$ of~$G$ maps each vertex $v \in V(G)$ to a distinct point in the plane and each edge $uv \in E(G)$ to a simple Jordan arc between the points corresponding to $u$ and~$v$. We always assume that $\Gamma$ is a \emph{simple} drawing, that is: $(i)$ two \emph{adjacent} edges (i.e., edges that share a vertex) do not intersect, except at their common endpoint (in particular, no edge is self-crossing); $(ii)$ two \emph{independent}  (i.e.\ non-adjacent) edges intersect at most in one of their interior points, called a \emph{crossing point}; $(iii)$ no three edges intersect at a common crossing point.

Let $\Gamma$ be a drawing of $G$. A \emph{vertex of~$\Gamma$} is either a point corresponding to a vertex of $G$, called a \emph{real-vertex}, or a point corresponding to a crossing point, called a \emph{crossing-vertex} or simply a \emph{crossing}.
We remark that in the literature a plane graph obtained by replacing crossing points with dummy vertices is often referred to as a \emph{planarization}~\cite{DBLP:books/ph/BattistaETT99}.
We denote by $V(\Gamma)$ the set of vertices of $\Gamma$. An \emph{edge} of~$\Gamma$ is a curve connecting two vertices of $\Gamma$. We denote by $E(\Gamma)$ the set of edges of $\Gamma$. An edge $e \in E(\Gamma)$ is a portion of an edge in $E(G)$, which we denote by $\overline{e}$; if both the endpoints of $e$ are real-vertices, then $e$ and $\overline{e}$ coincide. 

Drawing $\Gamma$ subdivides the plane into topologically connected regions, called \emph{faces}. The boundary of a face consists of a cyclical sequence of vertices (real- or crossing-vertices) and edges of $\Gamma$. We denote by $F(\Gamma)$ the set of faces of $\Gamma$. Exactly one face in~$F(\Gamma)$ corresponds to an infinite region of the plane, called the \emph{external face} of~$\Gamma$; the other faces are the \emph{internal faces} of~$\Gamma$.
If the boundary of a face $f$ of $\Gamma$ contains a vertex~$v$ (or an edge~$e$), we say that $f$ \emph{contains} $v$~(or~$e$).

In the following, if not specified, we denote by $n=|V(G)|$ and $m=|E(G)|$ the number of vertices and the number of edges of $G$, respectively. 
%For a drawing~$\Gamma$ of~$G$, we denote by $\nu = |V(\Gamma)|$, $\mu = |E(\Gamma)|$, and $\varphi = |F(\Gamma)|$ the number of vertices, edges, and faces of $\Gamma$, respectively. Also, we denote by $\chi = |V(\Gamma) \setminus V(G)| = \nu - n$ the number of crossing-vertices of $\Gamma$.
%% See \cref{tb:cardinality,tb:symbols} in \cref{app:glossary} for a summary of this notation and for a glossary of the symbols used in the~paper.

\smallskip\noindent {\bf Degree of vertices and faces.} 
For a vertex $v \in V(G)$, denote by $\deg_G(v)$ the \emph{degree of $v$ in $G$}, i.e., the number of edges incident to $v$. Analogously, for a vertex $v \in V(\Gamma)$, denote by $\deg_\Gamma(v)$ the \emph{degree of $v$ in $\Gamma$}. Note that, if $v \in V(G)$ then $\deg_\Gamma(v)=\deg_G(v)$, while if $v$ is a crossing-vertex then $\deg_\Gamma(v)=4$.   
For a face $f \in F(\Gamma)$, denote by $\deg_\Gamma(f)$ the \emph{degree of~$f$}, i.e., the number of times we traverse vertices (either real- or crossing-vertices) while walking on the boundary of~$f$ clockwise. Each vertex contributes to $\deg_\Gamma(f)$ the number of times we traverse it (possibly more than once if the boundary of~$f$ is not a simple cycle). 
%Note that, $\deg_\Gamma(f)$ can also be regarded as the number of edges of $\Gamma$ that are traversed while walking on the boundary of~$f$ clockwise (each edge counted with its multiplicity). 
Also, denote by $\deg_\Gamma^r(f)$ the \emph{real-vertex degree of~$f$}, i.e., the number of times we traverse a real-vertex of $\Gamma$ while walking on the boundary of $f$ clockwise. Again, each real-vertex contributes to $\deg_\Gamma^r(f)$ the number of times we traverse it. Finally, $\deg_\Gamma^c(f)$ denotes the number of times we traverse a crossing-vertex of $\Gamma$ while walking on the boundary of $f$ clockwise. Clearly, $\deg_\Gamma(f) = \deg_\Gamma^r(f) + \deg_\Gamma^c(f)$.

We say that a face $f \in F(\Gamma)$ is an \emph{$h$-real face}, for~$h \geq 0$, if $\deg_\Gamma^r(f)=h$. 
%Note that, if $\Gamma$ is a \kface graph, then each face of $\Gamma$ is an $h$-real face for some $h \geq k$.
An $h$-real face of degree $d$ is called an \emph{$h$-real $d$-gon}. For $k=2,3,4,5,6$, a face that is an $h$-real $k$-gon, is also called an \emph{$h$-real bigon} ($k=2)$, an \emph{$h$-real triangle} ($k=3)$, an \emph{$h$-real quadrilateral} ($k=4$), an \emph{$h$-real pentagon} ($k=5$), and an \emph{$h$-real hexagon} ($k=6$), respectively.  
An edge $e=uv \in E(\Gamma)$ is an \emph{$h$-real edge} ($h \in \{0,1,2\}$) if $|\{u,v\} \cap V(G)| = h$, i.e., $e$ contains $h$ real-vertices.

\smallskip\noindent {\bf Beyond-planar graphs.} 
A family $\cal F$ of \emph{beyond-planar graphs} is a set of (nonplanar) graphs that admit drawings with desired or forbidden edge-crossing configurations~\cite{DBLP:journals/csur/DidimoLM19}. The \emph{edge density} of a graph $G \in {\cal F}$ is the ratio between its number $m$ of edges and its number $n$ of vertices.  Graph $G$ is \emph{maximally dense} if it has the maximum edge density over all graphs of $\cal F$ with $n$ vertices. Graph $G$ is \emph{optimal} if it has the maximum edge density over all graphs in $\cal F$. Note that $\cal F$ might not contain optimal graphs for all values of $n$ (see, e.g., \cite{DBLP:journals/csur/DidimoLM19}).

%%%%%%%%%%%%%%%%%%%%%%%%%%%%%%%%%
%%%%%%%%% EDGE DENSITY %%%%%%%%%%
%%%%%%%%%%%%%%%%%%%%%%%%%%%%%%%%%

\section{Edge Density of Min-$k$-planar Graphs}\label{se:edge-density}

We start by proving some general bounds on the number of crossings in a min-k-planar drawing and on the number of edges of min-$k$-planar graphs.

\begin{property}\label{pr:crossings-min-k}
Any min-$k$-planar drawing $\Gamma$ of a graph $G$ (with $k \geq 1$) has at most $k \cdot \ell$ crossings, where $\ell$ is the number of light edges of $G$ in $\Gamma$.
\end{property}
\begin{proof} 
Two heavy edges cannot cross, thus each crossing in $\Gamma$ belongs to at least one light edge. Since each light edge has at most $k$ crossings, the bound follows.\end{proof}

\begin{property}\label{pr:heavy-min-k} 
Let~$\Gamma$ be a min-$k$-planar drawing of an $m$-edge graph $G$ (with $k \geq 1$). The number of heavy edges of $G$ in $\Gamma$ is at most $\frac{k}{2k+1} \cdot m$.
\end{property}
\begin{proof}
Let $h$ and $\ell$ be the number of heavy edges and the number of light  edges of $G$ in $\Gamma$, respectively. Observe that $m \geq h+\ell$. By definition, each heavy edge contains at least $(k+1)$ crossings, and two heavy edges do not cross. Hence, the number of crossings in $\Gamma$ is at least $h\cdot(k+1)$. By \cref{pr:crossings-min-k}, we have $h\cdot(k+1) \leq k \cdot \ell \leq k \cdot m - k \cdot h$, which implies $h \leq \frac{k}{2k+1} \cdot m$. 
\end{proof}

% \noindent Since the subgraph consisting of the heavy edges only is planar, we have:

% \begin{property}\label{pr:heavy-min-k-more}
% Let~$\Gamma$ be a min-$k$-planar drawing of an $n$-vertex graph $G$ (with $k \geq 1$). The subgraph of $G$ consisting of all edges that are not light edges in $\Gamma$ has at most $3n-6$ edges.  
% \end{property}

We now give a general bound on the edge density of min-$k$-planar simple graphs, for any $k \geq 2$. Finer bounds for $k=1,2,3$ are given in the next sections.

\begin{restatable}{theorem}{thdensityminkgeneral}\label{th:density-mink-general}
For any min-$k$-planar simple graph $G$ with $n$ vertices and $m$ edges it holds $m \leq \min\{5.39 \sqrt{k}\cdot n, (3.81\sqrt{k} + 3)\cdot n\}$ when $k \geq 2$.
\end{restatable}
\begin{proof}[Sketch]
Let $\mu = \min\{5.39 \sqrt{k} \cdot n, (3.81\sqrt{k} + 3)\cdot n\}$. Note that $\mu = 5.39 \sqrt{k} \cdot n$ when $2 \leq k \leq 3$, while $\mu = (3.81\sqrt{k} + 3)\cdot n$ when $k \geq 4$.

Suppose first that $2 \leq k \leq 3$. If $m < 6.95n$, the relation $m \leq 5.39\sqrt{k}\cdot n$ trivially holds. If $m \geq 6.95n$, let $\cross(G)$ be the minimum number of crossings required by any min-$k$-planar drawing $\Gamma$ of $G$. The improved version by Ackerman of the popular Crossing Lemma (Theorem~6 in \cite{ackerman:on}) implies that $\cross(G) \geq \frac{1}{29}\frac{m^3}{n^2}$. If $\ell$ is the number of light edges of $G$ in $\Gamma$, by \cref{pr:crossings-min-k} we have $\cross(G) \leq k \cdot \ell \leq k \cdot m$. Hence $\frac{1}{29}\frac{m^3}{n^2} \leq k \cdot m$, which yields $m \leq 5.39 \sqrt{k} \cdot n$.        

Suppose now that $k \geq 4$ and let $\Gamma$ be any min-$k$-planar drawing of $G$ with $\ell$ light edges. Since no two heavy edges cross, the subgraph of $G$ consisting of all heavy and crossing-free edges in $\Gamma$ has at most $3n-6$ edges, hence $m \leq \ell + 3n - 6$. Let $G'$ be the subgraph of $G$ consisting of the $\ell$ light edges of $G$ only. By applying Ackerman's version of the Crossing Lemma to $G'$, one can prove that $\ell \leq 3.81\sqrt{k}\cdot n$. Therefore, $m \leq \ell + 3n - 6 \leq \ell +3n \leq (3.81\sqrt{k} + 3) \cdot n$.        
\end{proof}

%In the next subsections we provide finer bounds for $k=1,2,3$.

%%%%% SUB EDGE-DENSITY k=1 %%%%%%%%
\subsection{Density of Min-1-planar graphs}\label{sse:edge-density-min-1}

Let~$\Gamma$ be a min-1-planar drawing of a graph $G$. We color each edge of $E(G)$ either red or green with the following rule: $(i)$ edges that are crossing-free in $\Gamma$ are colored red; $(ii)$ if $\{e_1,e_2\} \in E(G)$ is a pair of edges that cross in $\Gamma$, with $\cross(e_1) \geq \cross(e_2)$, we color $e_1$ as green and $e_2$ as red (if $\cross(e_1) = \cross(e_2) = 1$, the red edge is chosen arbitrarily). 
Note that, since $\Gamma$ is a min-1-planar drawing, each red edge is crossed at most once, hence the above coloring rule is well-defined. In particular, heavy edges are always colored green, while if two light edges cross, one is colored green and the other is colored red.
Hence, the subgraph induced by the red edges is a plane graph, called the \emph{red subgraph of $G$ defined by $\Gamma$}, or simply the \emph{red subgraph of $\Gamma$}. The following lemma is proved in \Cref{se:app-edge-density}.

\begin{restatable}{lemma}{ledensityminonesupport}\label{le:density-min1-support}
Let $G$ be a simple graph and let $\Gamma$ be a min-1-planar drawing of $G$. We can always augment $\Gamma$ with edges in such a way that the new drawing is still min-1-planar and all faces of its red subgraph  have degree three.  
\end{restatable}

We now prove a tight bound on the edge density of min-1-planar graphs.

\begin{restatable}{theorem}{thdensityminone}\label{th:density-min1}
Any $n$-vertex min-1-planar simple graph has at most $4n-8$ edges, and this bound is tight.
\end{restatable}
\begin{proof}
    Let $\Gamma$ be a min-1-planar drawing of a simple graph $G$ with $n$ vertices. By \cref{le:density-min1-support}, we can augment $\Gamma$ (and hence $G$) with new
    % crossing-free
    edges, in such a way that the new drawing $\Gamma'$ (and the corresponding graph $G'$) is min-1-planar and its red subgraph $\Gamma'_r$ is a triangulated planar graph. Hence, $\Gamma'_r$ has exactly $3n-6$ edges and $2n-4$ faces. 
    Every green edge of $G'$ (which is also a green edge of~$G$) traverses at least two faces of $\Gamma'_r$. Also, since $\Gamma'$ is a min-1-planar drawing and the red subgraph has only triangular faces, each face of the red subgraph is crossed by at most one green edge. Hence the number of green edges is at most  
    $\frac{2n - 4}{2} = n-2$, and therefore $G'$ has at most $(3n-6) + (n-2) = 4n-8$ edges in total.
    Since $G$ is a subgraph of $G'$, then also $G$ has at most $4n-8$ edges.
    
    About the tightness of the bound, we recall that optimal 1-planar graphs with $n$ vertices (which are also min-1-planar) have $4n-8$ edges \cite{bodendiek,DBLP:journals/combinatorica/PachT97,MR0187232}.
\end{proof}

Plugging the bound of \Cref{th:density-min1} 
into the bound of \Cref{pr:heavy-min-k}, we immediately get that any min-1-planar drawing has at most $\frac{4}{3}n-\frac{8}{3}$ heavy edges. We considerably improve this bound in the next theorem.

\begin{restatable}{theorem}{thheavyminone}\label{th:heavy-min1}
Any $n$-vertex min-1-planar drawing has at most $\frac{2}{3}n -1$~heavy~edges. Further, there exist min-1-planar drawings with $\frac{2}{3}n - O(1) $ heavy edges.
\end{restatable}
\begin{proof}
%We follow the lines of the proof of \cref{th:density-min1}.
    Let $\Gamma$ be a min-1-planar drawing of a simple graph $G$ with $n$ vertices. As in the proof of \cref{th:density-min1}, by \cref{le:density-min1-support} we can augment $\Gamma$ with new red edges, in such a way that the new drawing $\Gamma'$ is min-1-planar and its red subgraph $\Gamma'_r$ has all faces of degree three. 
    Hence, $\Gamma'_r$ has exactly $3n-6$ edges and $2n-4$ faces. Clearly, the number of heavy edges of $\Gamma'$ is not smaller than the one of $\Gamma$. 
    By definition, every heavy edge of $\Gamma'$ is crossed at least twice, hence it traverses at least three faces of $\Gamma'_r$. As before, each face of the red subgraph is crossed by at most one heavy edge. Hence the number of heavy edges is at most  
    $\frac{2n - 4}{3} \leq \frac{2}{3}n -1 $.
    For the lower bound, we refer to the left part of \cref{fi:min-k-heavy-edge-lower-bound} in the appendix.
\end{proof}

%%%%% SUB EDGE-DENSITY k=2 %%%%%%%%
\subsection{Density of Min-2-planar graphs}\label{sse:edge-density-min-2}

Proving a tight bound on the edge density of min-2-planar graphs is more challenging than for min-1-planar graphs. Observe that there are min-2-planar simple graphs with $5n-10$ edges, namely the optimal 2-planar graphs~\cite{DBLP:conf/compgeom/Bekos0R17}.
Each optimal 2-planar drawing consists of a subset of planar edges forming faces of size five (i.e., pentagons), and each face is filled up with five more edges that cross each other twice.
In the following we prove that 
$5n-10$ is also an upper bound to the number of edges of min-2-planar graphs.
To this aim, for any $k \geq 1$, we introduce a class of multi-graphs that generalize min-$k$-planar simple graphs.

\smallskip
Let $G$ be a (multi-)graph (without self-loops) and let $\Gamma$ be a (simple) drawing of $G$. A set of parallel edges of $G$ between the same pair of vertices is called a \emph{bundle} of $G$. We say that $\Gamma$ is \emph{bundle-proper} if for every bundle in $G$:
$(i)$ at most one of the edges of the bundle is involved in a crossing; and $(ii)$ $\Gamma$ has no face bounded only by two edges of the bundle (i.e., no face of $\Gamma$ is a 2-real bigon). 
We remark that, in the literature, two parallel edges that form a face of degree two are called \emph{homotopic}. Hence, property $(ii)$ is equivalent to saying that a \emph{bundle-proper} drawing does not contain homotopic parallel edges.

Graph $G$ is \emph{bundle-proper min-$k$-planar} if it admits a (simple) drawing $\Gamma$ that is both min-$k$-planar and bundle-proper. If $G$ has $n$ vertices and has the maximum number of edges over all bundle-proper min-$k$-planar $n$-vertex graphs, then we say that $G$ is a \emph{maximally-dense} bundle-proper min-$k$-planar graph. Consider a pair $(G,\Gamma)$, where 
$G$ is an $n$-vertex bundle-proper min-$k$-planar graph and $\Gamma$ is a bundle-proper min-$k$-planar drawing of $G$. We say that $(G,\Gamma)$ is a \emph{maximally-dense crossing-minimal bundle-proper min-$k$-planar pair} if $G$ is maximally-dense and $\Gamma$ has the minimum number of crossings over all bundle-proper min-$k$-planar drawings of maximally-dense bundle-proper min-$k$-planar $n$-vertex graphs. The proof of the next lemma is in \Cref{se:app-edge-density}.

\begin{restatable}{lemma}{lemaximalmintwoplanarprop}\label{le:maximal-min-2-planar-prop}
  Let $(G,\Gamma)$ be a maximally-dense crossing-minimal bundle-proper min-$k$-planar pair. These properties hold:
  \begin{inparaenum}[$(a)$]
    %\item No edge of $G$ crosses itself.\label{no_self_crossing}
    \item If a face~$f$ of $\Gamma$ contains two distinct real-vertices $u$ and $v$, then $f$ contains an edge $uv$.\label{edge_appears}
%
    %\item No two edges of $G$ cross more than once.\label{no_double_crossing}
%
    \item For each face~$f$ of $\Gamma$,  $\deg_{\Gamma}(f) \geq 3$.\label{no_bigon}
    \item A face $f$ of $\Gamma$ with $\deg_{\Gamma}^r(f) \geq 3$ is a 3-real triangle.\label{3_real_triangle}
    
    %\item If $e_1,e_2,e_3$ are the three edges of a $0$-real triangle in $\Gamma$, then the edges $\overline{e}_1, \overline{e}_2, \overline{e}_3$ of $G$ are pairwise independent. \label{six_endpoints}
  \end{inparaenum}
\end{restatable}

To prove the upper bound we use \emph{discharging} techniques. See~\cite{ackerman:on,ackerman.tardos:on,bdd-23-wg,DBLP:journals/cjtcs/DujmovicGMW11} for previous works that use this tool.  
Define a \emph{charging function} $\ch: F(\Gamma) \rightarrow \mathbb{R}$ such that, for each $f \in F(\Gamma)$:
\begin{equation}\label{eq:initial-charge}
 \ch(f) = \deg_\Gamma(f) + \deg^r_\Gamma(f) - 4 = 2\deg^r_\Gamma(f) + \deg^c_\Gamma(f) - 4
\end{equation}
\noindent The value $\ch(f)$ is called the \emph{initial charge} of $f$. Using Euler's formula, it is not difficult to see that the following equality holds (refer to~\cite{ackerman.tardos:on} for details):
\begin{equation}\label{eq:initial-charge-sum}
    \sum_{f \in F(\Gamma)} \ch(f) = 4n - 8
\end{equation}
The goal of a discharging technique is to derive from the initial charging function $\ch(\cdot)$ a new function $\ch'(\cdot)$ that satisfies two properties: {\textsf{(C1)}} $\ch'(f) \geq \alpha \deg_\Gamma^r(f)$, for some real number $\alpha > 0$; and
{\textsf{(C2)}} $\sum_{f \in F(\Gamma)} \ch'(f) \leq \sum_{f \in F(\Gamma)} \ch(f)$. 

\smallskip
\noindent If $\alpha > 0$ is a number for which a function $\ch'(\cdot)$ satisfies $\textsf{(C1)}$ and $\textsf{(C2)}$, by \cref{eq:initial-charge-sum} we get:
$4n-8 = \sum_{f \in F(\Gamma)} \ch(f) \geq \sum_{f \in F(\Gamma)} \ch'(f) \geq \alpha \sum_{f \in F(\Gamma)} \deg_\Gamma^r(f)$.
Also, since $\sum_{f \in F(\Gamma)} \deg_\Gamma^r(f) = \sum_{v \in V(G)} \deg_G(v) = 2m$, we get the following:
\begin{equation}\label{eq:modified-charge}
    m \leq \frac{2}{\alpha}(n-2)
\end{equation}
Thus, \cref{eq:modified-charge} can be exploited to prove upper bounds on the edge density of a graph for specific values of $\alpha$, whenever we find a charging function $\ch'(\cdot)$ that fulfills \textsf{(C1)} and \textsf{(C2)}. We prove the following.

%upper bound edge-desntiy min-2-planar
%-------------------------------------

\begin{restatable}{theorem}{thdensitymintwo}\label{th:density-min2}
Any $n$-vertex min-2-planar simple graph has at most $5n-10$~edges, and this bound is tight. 
\end{restatable}
\begin{proof}[Sketch]
  We already observed that there exist min-2-planar simple graphs with $5n-10$ edges (e.g., the optimal 2-planar graphs). We now prove that min-2-planar simple graphs have at most $5n-10$ edges. Since any simple graph is also bundle-proper, we can show that the bound holds more in general for bundle-proper min-2-planar. Also, we can restrict to maximally-dense bundle-proper min-2-planar graphs, and in particular to crossing-minimal drawings. Let $(G,\Gamma)$ be any maximally-dense crossing-minimal bundle-proper min-2-planar pair, with $|V(G)|=n$.
  We show the existence of a charging function $\ch'(\cdot)$ that satisfies \textsf{(C1)} and \textsf{(C2)} for $\alpha = \frac{2}{5}$, so the result will follow from \cref{eq:modified-charge}.

  Consider the initial charging function $\ch(\cdot)$ defined in \cref{eq:initial-charge}. For each type of triangle $t$ we analyze the value of $\ch(t)$ and the deficit/excess w.r.t. $\frac{2}{5}\deg_\Gamma^r(t)$.
  If $t$ is a $0$-real triangle, $\ch(t)=-1 < 0 = \tfrac{2}{5}\deg_{\Gamma}^r(t)$, thus $t$ has a deficit of $1$. 
  If $t$ is a $1$-real triangle, $\ch(t)=0 < \tfrac{2}{5}=\tfrac{2}{5}\deg_{\Gamma}^r(t)$, thus $t$ has a deficit of $\frac{2}{5}$.
  If $t$ is a $2$-real triangle, $\ch(t)=1 > \tfrac{4}{5}=\tfrac{2}{5}\deg_{\Gamma}^r(t)$, thus $t$ has an excess of $\frac{1}{5}$.
  If $t$ is a $3$-real triangle, $\ch(t) = 2 > \tfrac{6}{5} = \tfrac{2}{5}\deg_{\Gamma}^r(t)$, thus $t$ has an excess of $\frac{4}{5}$.
  
  Also, if $f$ is any face of $\Gamma$ with $\deg_{\Gamma}(f) \geq 4$, then $\ch(f)=2\deg_{\Gamma}^r(f)+\deg_{\Gamma}^c(f)-4 = \deg_{\Gamma}(f)-4 + \deg_{\Gamma}^r(f) \geq \deg_{\Gamma}^r(f) \geq \tfrac{2}{5}\deg_{\Gamma}^r(f)$.
  Therefore $\ch(\cdot)$ only fails to satisfy \textsf{(C1)} at $0$-real and $1$-real triangles.  We begin by setting $\ch'(f)=\ch(f)$ for each face $f$ of $\Gamma$ and we explain how to modify $\ch'(\cdot)$ in such a way that $\ch'(f) \geq \frac{2}{5}\deg_{\Gamma}^r(f)$ for each face $f \in F(\Gamma)$, thus satisfying \textsf{(C1)}, and such that the total charge remains the same, thus satisfying \textsf{(C2)}. 
  
  \paragraph{Fixing $0$-real triangles.}
  Let $t$ be a $0$-real triangle in $\Gamma$ with edges $e_1$, $e_2$, and $e_3$.  Refer to \cref{fi:0-real-triangles}.  The edges $\overline{e}_1$, $\overline{e}_2$ and $\overline{e}_3$ are three pairwise crossing edges of $G$. Since $\Gamma$ is a simple drawing, $\overline{e}_1$, $\overline{e}_2$ and $\overline{e}_3$ are independent edges of $G$ (i.e., their six end-vertices are all distinct). Also, since $\Gamma$ is min-2-planar, at least two of these three edges, say $\overline{e}_2$ and $\overline{e}_3$, do not cross other edges of $G$ in $\Gamma$. This implies that each of the two end-vertices of $\overline{e}_2$ shares a face with an end-vertex of $\overline{e}_3$. Hence, by      
  \cref{le:maximal-min-2-planar-prop}(\ref{edge_appears}), the four vertices of $\overline{e}_2$ and $\overline{e}_3$ form a $4$-cycle $e'\overline{e}_2e''\overline{e}_3$ in $G$ and $\Gamma$ contains a $2$-real quadrilateral $f_1$ bounded by portions of $e''$, $\overline{e}_1$, $\overline{e}_2$, $\overline{e}_3$, and a $2$-real triangle $f_2$ bounded by portions of $e'$, $\overline{e}_2$, $\overline{e}_3$. 

\begin{figure}[tb]
    \centering
    \subfigure[]{\includegraphics[page=1,height=.2\textwidth]{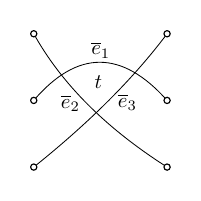}\label{fi:0-real-triangles-a}}
    \hfil
    \subfigure[]{\includegraphics[page=2,height=.2\textwidth]{0-real-triangles}\label{fi:0-real-triangles-b}}
    \hfil
    \subfigure[]
    {\includegraphics[page=3,height=.2\textwidth]{0-real-triangles}\label{fi:0-real-triangles-c}}
    \hfil
    \subfigure[]{\includegraphics[page=4,height=.2\textwidth]{0-real-triangles}\label{fi:0-real-triangles-d}}
    \caption{(a) A 0-real triangle $t$. (b) A 2-real quadrilateral $f_1$ and a 2-real triangle $f_2$ neighboring $t$. (c) The initial charges. (d) The charges after a redistribution.}
    \label{fi:0-real-triangles}
\end{figure}
  
  The charge of $f_1$ is $\ch'(f_1)=2$, with an excess of $\frac{6}{5}$ w.r.t. $\frac{2}{5}\deg_\Gamma^r(f_1)=\frac{4}{5}$. The charge of $f_2$ is $\ch'(f_2)=1$, with an excess of $\frac{1}{5}$ w.r.t. $\frac{2}{5}\deg_\Gamma^r(f_2)=\frac{4}{5}$. We reduce $\ch'(f_1)$ by $\frac{4}{5}$, reduce $\ch'(f_2)$ by $\frac{1}{5}$, and increase $\ch'(t)$ by $1$.
  After that, the total charge is unchanged and all the three faces $t$, $f_1$, and $f_2$ satisfy \textsf{(C1)}. Namely, $\ch'(t)=0$ (it has no deficit/excess), $\ch'(f_1)=\frac{6}{5}$ (it has an excess of $\frac{2}{5}$), and $\ch'(f_2)=\frac{4}{5}$ (it has no deficit/excess).   
  In the remainder of the proof, we call each of the faces $f_1$ and $f_2$ a \emph{$0$-real triangle-neighboring face}.  Each $0$-real triangle-neighboring face that is a $2$-real triangle (as $f_2$) shares its unique crossing vertex with a $0$-real triangle; each $0$-real triangle-neighboring face that is a $2$-real quadrilateral (as $f_1$) shares its unique $0$-real edge with a $0$-real triangle.
  
  \paragraph{Fixing $1$-real triangles.}
  Let $t$ be a $1$-real triangle, with real-vertex $v_1$ and crossing-vertices $v_2$ and $v_3$.  Refer to \cref{fi:1-real-triangles} for an illustration. Let $e_0=v_2v_3$ be the $0$-real edge of $t$, and let $f_1$ be the face of $\Gamma$ that shares $e_0$ with $t$. If $f_1$ is a $0$-real quadrilateral, denote by $e_1$ the $0$-real edge of $f_1$ not adjacent to $e_0$, and by $f_2$ the face of $\Gamma$ that shares $e_1$ with $f_1$. If $f_2$ is a $0$-real quadrilateral, denote by $e_2$ the $0$-real edge of $f_2$ not adjacent to $e_1$, and by $f_3$ the face of $\Gamma$ that shares $e_2$ with $f_2$. We continue in this way until we encounter a face $f_p$ $(p \geq 1)$ that is not a $0$-real quadrilateral. This procedure determines a sequence of faces $f_0, f_1, f_2, \dots f_p$, and a sequence of $0$-real edges $e_0, e_1, \dots, e_{p-1}$ such that $f_0=t$, $f_i$ is a $0$-real quadrilateral for each $i \in \{1, \dots, p-1\}$, $f_p$ is not a $0$-real quadrilateral, and the faces $f_i$ and $f_{i-1}$ share edge $e_{i-1}$ $(i \in \{1, \dots, p\})$. 

  \begin{figure}[tb]
    \begin{center}
    \includegraphics[page=1,height=.25\textwidth]{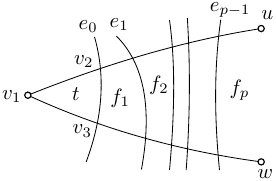}
    \end{center}
    \caption{The demand path of a $1$-real triangle $t$, ending at a face $f_p$.}
    \label{fi:1-real-triangles}
  \end{figure}
  
  Note that $\deg_\Gamma(f_p) \geq 4$. Namely, let $e=v_1v_2$ and $e'=v_1v_3$, and let $\overline{e}=v_1u$ and $\overline{e'}=v_1w$ be the edges of $G$ that contain $e$ and $e'$. Since $f_p$ has at least two crossing-vertices, if $f_p$ were a triangle then it would be either a $0$-real triangle or a $1$-real triangle.
  If $f_p$ were a $0$-real triangle then $\overline{e}$ and $\overline{e'}$ would cross in $\Gamma$, which is impossible as $\overline{e}$ and $\overline{e'}$ are adjacent edges and $\Gamma$ is a simple drawing. If $f_p$ were a $1$-real triangle then $u=w$, i.e., $\overline{e}$ and $\overline{e'}$ would be parallel edges both involved in a crossing, which is impossible as $\Gamma$ is bundle-proper. 
  Hence, $\deg_\Gamma(f_p) \geq 4$ and, as observed at the beginning of this proof, $\ch'(f_p) \geq \frac{2}{5} \deg_\Gamma^r(f_p)$. Also, the charge excess of $f_p$ is larger than $\frac{2}{5}$. Namely, this excess is $x = 2\deg_\Gamma^r(f_p) + \deg_\Gamma^c(f_p)-4 - \frac{2}{5}\deg_\Gamma^r(f_p) = \deg_\Gamma(f)  + \frac{3}{5}\deg_\Gamma^r(f) - 4$.
  If $f_p$ has no real-vertices, it must have at least five crossing-vertices (as $f_p$ is not a $0$-real quadrilateral), which implies $x \geq 1 > \frac{2}{5}$. If $f_p$ has at least one real-vertex then $x \geq \frac{3}{5} > \frac{2}{5}$.
  
  Therefore, the idea is to fill the $\frac{2}{5}$ charge deficit of $t$ by moving an equivalent amount of charge from $f_p$ to $t$. We say that  
  $t$ \emph{demands from $f_p$ through edge $e_{p-1}$} a charge amount of $\tfrac{2}{5}$. We call $f_0,\ldots,f_p$ (which is a path in the dual of $\Gamma$) the \emph{demand path} for $t$. Hence, for each $1$-real triangle $t$ of $\Gamma$ whose demand path ends at a face $f=f_p$, we decrease $\ch'(f)$ by $\tfrac{2}{5}$ and increase $\ch'(t)$ from $0$ to $\tfrac{2}{5}$. Note that $f$ cannot be a $0$-real triangle-neighboring face. Indeed, $f$ is not a triangle, and if $f$ is a $2$-real quadrilateral then its $0$-real edge is shared either with a $0$-real quadrilateral or directly with the $1$-real triangle $t$. It follows that the set of faces whose charge is affected by fixing $1$-real triangles does not intersect with the set of faces whose charge is affected by fixing $0$-real triangles.
  
  From the reasoning above, after we have fixed all 1-real triangles, we may have problems only if multiple $1$-real triangles demanded from the~same~face~$f$. In this case, $f$ might no longer satisfy \textsf{(C1)}. 
  The rest of the proof (see \Cref{se:app-edge-density}) shows which faces may be in this situation and how to fix their charge.
\end{proof}

Combining 
\Cref{th:density-min2} with \Cref{pr:heavy-min-k} we immediately get that any min-2-planar drawing has at most $2n-4$ heavy edges.
The next theorem considerably improves this bound by exploiting discharging techniques (see \Cref{se:app-edge-density}).

\begin{restatable}{theorem}{thheavymintwo}\label{th:heavy-min2}
Any $n$-vertex min-2-planar drawing has at most $\frac{6}{5}(n-2)$ heavy edges. Further, there exist min-2-planar drawings with $n-O(1)$ heavy edges.
\end{restatable}

%%%%% SUB EDGE-DENSITY k=3 %%%%%%%%
\subsection{Density of Min-3-planar graphs}\label{sse:edge-density-min-3}

For the family of min-3-planar graphs we consider graphs that can contain non-homotopic parallel edges. Indeed, it is known that $n$-vertex 3-planar graphs that are simple have at most $5.5 n - 15$ edges~\cite{DBLP:journals/combinatorica/PachT97}, but this bound is not tight. On the other hand, a tight upper bound is known for 3-planar graphs that can contain non-homotopic multiple edges, namely $5.5 n - 11$~\cite{DBLP:journals/dcg/PachRTT06}.    
We give upper bounds on the edge density and on the density of the heavy edges in min-3-planar graphs. 
%Clearly, this bound also works for simple graphs.
The proofs still exploit discharging techniques (see \Cref{se:app-edge-density}).
%Correspondingly, we derive upper bounds on the maximum number of heavy edges in a min-3-planar drawing.

\begin{restatable}{theorem}{thdensityminthree}\label{th:density-min3}
Any min-$3$-planar graph with $n$ vertices has at most $6n-12$ edges. 
\end{restatable}

%As for min-1-planar graphs, combining 
%\Cref{th:density-min2} with \Cref{pr:heavy-min-k} we immediately get that any min-2-planar drawing has at most $2n-4$ heavy edges.
%The next theorem considerably improves this bound. The proof still exploits discharging techniques (see \Cref{se:app-edge-density}).

\begin{restatable}{theorem}{thheavyminthree}\label{th:heavy-min3}
Any $n$-vertex min-3-planar drawing has at most $2(n-2)$ heavy edges. Further, there exist min-3-planar drawings with $\frac{6}{5}n-O(1)$ heavy edges.
\end{restatable}

%%%%%%%%%%%%%%%%%%%%%%%%%%%%%%%%%%%%%
%%%%%% INCLUSION RELATIONSHIPS %%%%%%
%%%%%%%%%%%%%%%%%%%%%%%%%%%%%%%%%%%%%

\section{Relationships with $k$-planar Graphs}\label{se:inclusion-relationships}

% Although the maximum edge densities of min-$k$-planar graphs and min-$k$-planar graphs coincide for $k=1,2$, in the following we show that for $k=1,2,3$, the classes of $k$-planar graphs are proper subclasses of the min-$k$-planar graphs. 
% %We also give some relationships with other popular beyond-planar graph families. 
%When not specified, we assume that the graphs are simple.
%
%%%%% relationships k=1 %%%%%%%%
%\smallskip\noindent{\bf Min-1-planar graphs}. 
The next theorem shows that while the family of min-1-planar graphs properly contains the family of 1-planar graphs, the two classes coincide when we restrict to optimal graphs, i.e., those with $4n-8$ edges.

\begin{restatable}{theorem}{thminoneoneplanar}\label{th:min1-1planar}
1-planar graphs are a proper subset of min-1-planar graphs, while
optimal min-1-planar graphs are optimal 1-planar.
\end{restatable}
\begin{proof}
Any 1-planar graph is min-1-planar. By the \NP-hardness of testing whether a given planar graph plus a single edge is 1-planar~\cite{DBLP:journals/siamcomp/CabelloM13}, we know that there are such graphs that are not 1-planar, while any planar graph that is extended by a single edge can be drawn min-1-planar. Hence, 1-planar graphs are a proper subset of min-1-planar graphs.  
Finally, as in the proof of \cref{th:density-min1}, in every optimal min-1-planar drawing the red subgraph is maximal planar and each green edge traverses exactly two faces of the red subgraph. Hence, each green edge crosses exactly once, i.e., the drawing is also (optimal) 1-planar.
\end{proof}

%%%%% relationships k=2 %%%%%%%%
%\smallskip\noindent{\bf Min-2-planar graphs.} 
Unlike min-1-planar graphs, we show that min-2-planar graphs are a proper superset of the 2-planar graphs even when we restrict to optimal graphs.  

\begin{restatable}{theorem}{thmintwotwoplanar}\label{th:min2-2planar}
2-planar graphs are a proper subset of min-2-planar graphs,
and there are optimal min-2-planar graphs that are not optimal 2-planar.
\end{restatable}
\begin{proof}[Sketch]
First observe that there exist non-optimal min-2-planar graphs that are not 2-planar. For example, $K_{5,5}$ is not 2-planar~\cite{DBLP:conf/gd/AngeliniB0S19,DBLP:journals/jgaa/AngeliniBKS20},  
while \cref{fi:K55} illustrates a min-2-planar drawing of $K_{5,5}$. 
% \begin{figure}[tb]
%     \centering
%     \includegraphics[width=.4\textwidth]{K55}
% \caption{A min-2-planar drawing of the complete bipartite graph $K_{5,5}$.}\label{fi:K55}
% \end{figure}
%
%In the following, we show how to construct optimal $n$-vertex min-2-planar graphs that are not 2-planar.
%
\begin{figure}[tb]
    \centering
    \subfigure[]{\includegraphics[page=1,width=.3\textwidth]{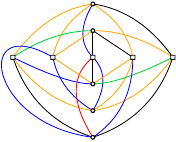}\label{fi:K55}}
    \hfil
    \subfigure[]{\includegraphics[page=1,width=.34\textwidth]{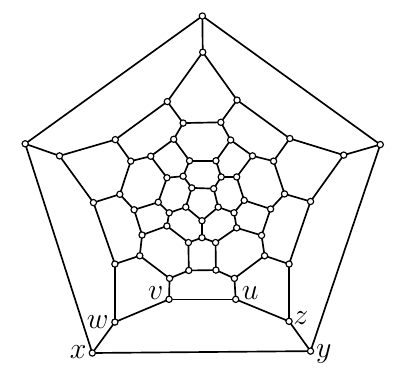}\label{fi:truncated-icosahedral-graph-a}}
    \hfil
    \subfigure[]{\includegraphics[page=2,width=.34\textwidth]{truncated-icosahedral-graph}\label{fi:truncated-icosahedral-graph-b}}
    \caption{
    (a) A min-2-planar drawing of $K_{5,5}$. (b) A planar drawing $\Gamma$ of the truncated icosahedral graph $G$. (c) A min-2-planar drawing $\Gamma'$ of the graph $G'$, obtained by adding 5 edges to each pentagonal face and 7 edges to each hexagonal face of $\Gamma$.}
    \label{fi:truncated-icosahedral-graph}
\end{figure}
To construct an optimal min-2-planar graph that is not 2-planar, start from the truncated icosahedral graph in \cref{fi:truncated-icosahedral-graph-a}, consisting of 12 pentagonal faces, 20 hexagonal faces, 60 vertices and 90 edges. Then enrich it with 5 edges inside each pentagonal face and 7 edges inside each hexagonal face, as in \cref{fi:truncated-icosahedral-graph-b}. The new graph is min-2-planar and has $5n-10$ edges. Conversely, it is not 2-planar as it contains vertices of degree 10 and 11, while it is known that optimal 2-planar graphs must have only vertices with degree multiple of three~\cite{DBLP:conf/gd/Forster0R21}. See \Cref{se:app-inclusion-relationships} for details.
\end{proof}

%%%%% relationships k=3 %%%%%%%%
%\smallskip\noindent{\bf Min-3-planar graphs.}
In contrast to 1- and 2-planar graphs, the maximum densities of 3-planar and min-3-planar graphs differ.
%This implies in particular, that not all (optimal) min-3-planar graphs are (optimal) 3-planar graphs.

\begin{restatable}{theorem}{thminthreethreeplanar}\label{th:min3-3planar}
% There is a non-simple min-$3$-planar graph that has strictly more edges than any non-simple $3$-planar graphs.
% 3-planar graphs are a proper subset of min-3-planar graphs, and 
There are min-3-planar (non-simple) graphs denser than optimal 3-planar (non-simple) graphs.
\end{restatable}
\begin{proof}[Sketch]
\begin{figure}[tb]
    \centering
    \subfigure[]{\includegraphics[page=2,width=.4\textwidth]{eight-gons}\label{fi:8-gons-a}}
    \hfil
    \subfigure[]{\includegraphics[page=1,width=.4\textwidth]{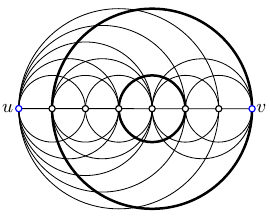}\label{fi:8-gons-b}}
    \caption{
    Illustration of the construction of \Cref{th:min3-3planar}. If in the graph of figure (a) we replace each shaded chain with a copy of the graph of figure (b), we get a min-3-planar graph that is not 3-planar. The bold edges are heavy edges.}
    \label{fi:8-gons}
\end{figure}
First, consider a planar graph $G$, and a corresponding drawing $\Gamma$, consisting of $h$ parallel chains ($h \geq 1$), each with $8$ vertices, sharing the two end-vertices $u$ and $v$, and interleaved by $h$ copies of edge $uv$; refer to \cref{fi:8-gons-a}. Then, construct a new graph $G'$, and a corresponding drawing $\Gamma'$, obtained from~$G$, and from $\Gamma$, by replacing each parallel chain with a copy of the graph $G''$ depicted in \cref{fi:8-gons-b}. In the drawing $\Gamma'$, each copy of $G''$ has the same edge crossings as the drawing illustrated in \cref{fi:8-gons-b}.
Graph $G''$ has $8$ vertices and $33$ edges, and it is min-3-planar. The bold edges are the heavy edges in the drawing of \cref{fi:8-gons-b}). It can be proved that $G''$ has $5.\overline{6}n-11.\overline{3}$ edges, while it is known that every 3-planar graph has at most $5.5n-11$ edges~\cite{DBLP:journals/dcg/PachRTT06}. 
\end{proof}

\section{Final Remarks and Open Problems}\label{se:open}

About edge density, one can ask whether the bound of \Cref{th:density-min3} for min-3-planar graphs is tight or if it can be further lowered. Providing finer bounds for $k \geq 4$ is also interesting.
Another classical research direction is to establish inclusion or incomparability relations between min-$k$-planar graphs and classes of beyond-planar graphs other than $k$-planar graphs. 
The next two lemmas provide initial results in this direction (see \Cref{se:app-conclusions} for their proofs). In particular, \cref{le:mink-gapplanar} leaves as open what is the relationship between min-2-planar graphs and 1-gap-planar graphs (which have the same maximum edge density). \Cref{le:mink-fanplanar} implies that min-2-planar graphs and fan-planar graphs are incomparable classes, even if they have the same maximum edge density.
% In this direction, the next lemma highlights interesting, simple to prove, inclusion relationships between min-$k$-planar graphs and two popular families in beyond-planar graph drawing (see \Cref{se:app-conclusions}). It leaves as open the questions whether min-2-planar graphs are 1-gap-planar.
% Also, \Cref{le:mink-fanplanar} implies that min-2-planar graphs and fan-planar graphs are incomparable classes, even if they have the same maximum edge density.

% \begin{restatable}{lemma}{leminonegapplanar}\label{le:min1-gapplanar}
% Min-1-planar graphs are a proper subset of $k$-gap-planar graphs and of $(k+2)$-quasiplanar graphs, for every $k \geq 1$.
% \end{restatable}
% \begin{proof}
% It suffices to prove the statement for $k=1$.
% Every min-1-planar drawing can be made  1-gap-planar by introducing a gap on every red edge crossing a green edge (see the coloring rule in \cref{sse:edge-density-min-1}). Indeed, by definition, each red edge receives only one gap.
% Finally, if a min-1-planar drawing is not 3-quasi planar, then it has three mutually crossing edges, which implies the existence of two edges that cross at least two times each.
% Also, since optimal 1-gap-planar graphs have $5n-10$ edges \cite{DBLP:journals/tcs/BaeBCEE0HKMRT18} and there exist 3-quasiplanar graphs with $6.5n-O(1)$ edges \cite{ackerman.tardos:on}, the inclusion relationships are proper.
% \end{proof}

\begin{restatable}{lemma}{leminkgapplanar}\label{le:mink-gapplanar}
Min-$k$-planar graphs are a subset of $k$-gap-planar graphs and of $(k+2)$-quasiplanar graphs, for every $k \geq 1$.
\end{restatable}

% \Cref{le:mink-fanplanar} provides another interesting result, showing that min-$k$-planar graphs and fan-planar graphs are incomparable classes. This implies in particular the incomparability between min-2-planar graphs and fan-planar graphs, even if they have the same maximum edge density (see \Cref{se:app-conclusions}).

\begin{restatable}{lemma}{leminkfanplanar}\label{le:mink-fanplanar}
For any given $k \geq 2$, fan-planar and min-$k$-planar graphs are incomparable, i.e.,  each of the two classes contains graphs that are not in the other.
\end{restatable}
\bibliographystyle{splncs04}
\bibliography{biblio}

%%%%%%%%%%%%%%%%%%%% APPENDIX %%%%%%%%%%%%%%%%%%
\newpage
\appendix

\section{Appendix}

\subsection{Details for \Cref{se:edge-density}}\label{se:app-edge-density}

\thdensityminkgeneral* 
\begin{proof}
Let $\mu = \min\{5.39 \sqrt{k} \cdot n, (3.81\sqrt{k} + 3)\cdot n\}$. Note that $\mu = 5.39 \sqrt{k} \cdot n$ when $2 \leq k \leq 3$, while $\mu = (3.81\sqrt{k} + 3)\cdot n$ when $k \geq 4$. Hence, we prove that for $2 \leq k \leq 3$ we have $m \leq 5.39 \sqrt{k} \cdot n$, while for $k \geq 4$ we have $m \leq (3.81\sqrt{k} + 3)\cdot n$.

Suppose first that $2 \leq k \leq 3$. If $m < 6.95n$, the relation $m \leq 5.39\sqrt{k}\cdot n$ trivially holds, as $5.39 \sqrt{k} \cdot n \leq 7.63n$. If $m \geq 6.95n$, let $\cross(G)$ be the minimum number of crossings required by any min-$k$-planar drawing $\Gamma$ of $G$. The improved version by Ackerman of the popular Crossing Lemma (Theorem~6 in \cite{ackerman:on}) implies that $\cross(G) \geq \frac{1}{29}\frac{m^3}{n^2}$. If $\ell$ is the number of light edges of $G$ in $\Gamma$, by \cref{pr:crossings-min-k} we have $\cross(G) \leq k \cdot \ell \leq k \cdot m$. Hence $\frac{1}{29}\frac{m^3}{n^2} \leq k \cdot m$, which yields $m \leq 5.39 \sqrt{k} \cdot n$.        

Suppose now that $k \geq 4$ and let $\Gamma$ be any min-$k$-planar drawing of $G$ with $\ell$ light edges. Since no two heavy edges cross, the subgraph of $G$ consisting of all heavy and crossing-free edges in $\Gamma$ has at most $3n-6$ edges, hence $m \leq \ell + 3n - 6$. Let $G'$ be the subgraph of $G$ consisting of the $\ell$ light edges of $G$ only, and let $\Gamma'$ be the restriction of $G'$ in $\Gamma$. We show that $\ell \leq 3.81\sqrt{k}\cdot n$. The relation trivially holds when $\ell < 6.95 n$, as $k \geq 4$. If $\ell \geq 6.95n$, using  Ackerman's version of the Crossing Lemma applied to $G'$, we have $\cross(G') \geq \frac{1}{29}\frac{\ell^3}{n^2}$. Also, $\Gamma'$ has at most $\frac{k\cdot\ell}{2}$ crossings, because each light edge has at most $k$ crossings and each crossing is shared by two edges of $G'$. It follows that $\frac{1}{29}\frac{\ell^3}{n^2} \leq \frac{k\cdot\ell}{2}$, which still implies $\ell \leq 3.81\sqrt{k}\cdot n$. Therefore, $m \leq \ell + 3n - 6 \leq \ell +3n \leq (3.81\sqrt{k} + 3) \cdot n$.        
\end{proof}

\ledensityminonesupport*
\begin{proof}
    Let $\Gamma_r$ be the red subgraph of $\Gamma$. Since $G$ is simple, every face of $\Gamma$ has degree greater than two.
    Suppose that $\Gamma_r$ has at least one face $f$ such that $\deg_{\Gamma_r}(f) \geq 4$. We augment $\Gamma$ with new red edges in two steps, described below. The augmentation may introduce multiple edges, but it will guarantee that all faces of the new red subgraph have degree three.
    
\begin{figure}[b]
	\centering
    \subfigure[]{\includegraphics[page=1,width=.19\textwidth]{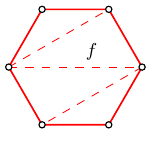}\label{fi:augmentation-a}}
    \hfil
    \subfigure[]{\includegraphics[page=10,width=.19\textwidth]{augmentation}\label{fi:augmentation-b}}
    \hfil
    \subfigure[]{\includegraphics[page=5,width=.19\textwidth]{augmentation}\label{fi:augmentation-c}}
    \hfil
    \subfigure[]{\includegraphics[page=9,width=.19\textwidth]{augmentation}\label{fi:augmentation-d}}
    % \hfil
    % \subfigure[]{\includegraphics[page=8,width=.19\textwidth]{augmentation}\label{fi:augmentation-e}}
    \caption{
    Illustration for the proof of \cref{le:density-min1-support}.}
    \label{fi:augmentation}
\end{figure}
    
    \begin{compactitem}
    	\item[\textsf{Step 1.}] Suppose that there exists a face $f$ of $\Gamma_r$ with $\deg_{\Gamma_r}(f) \geq 4$ and containing two vertices $u$ and $v$ that can be connected by an edge $uv$ that splits $f$ without crossing other edges of $\Gamma$. We add edge $uv$ and color it as red (as it is crossing-free); we also say that this operation  \emph{augments} $f$. We repeat this procedure until no such a face $f$ exists. The obtained drawing is still min-1-planar, since we added only crossing-free edges.
     
    	\item[\textsf{Step 2.}] Suppose that $\Gamma_r$ still contains a face $f$ with $\deg_{\Gamma_r}(f) \geq 4$.
    	Observe that:
    	
    	$(i)$ Face $f$ is traversed by a green edge in $\Gamma$, otherwise it would have been augmented in \textsf{Step 1}; see \cref{fi:augmentation-a}.

    	$(ii)$ Every green edge $e$ that traverses $f$ is not incident to any vertex $u$ of $f$. Namely, suppose for contradiction that $e$ is incident to a vertex $u$ of $f$, and let $e_r=vw$ be the red edge of $f$ crossed by $e$. Since $\deg_{\Gamma_r}(f) \geq 4$, at least one among $v$ and $w$, say for example $v$, is not adjacent to $u$. However this implies that either $f$ can be augmented by adding a red edge $uv$, which contradicts that we completed \textsf{Step~1}, or there is another green edge that crosses $e_r$, which contradicts that $e_r$ is crossed at most once; see \cref{fi:augmentation-b}.    
    
        $(iii)$ Face $f$ cannot be traversed by two distinct green edges $e_1$ and $e_2$ (refer to \cref{fi:augmentation-c}). More precisely, if this is the case, these edges cannot cross each other and, by property $(ii)$, each of $e_1$ and $e_2$ crosses two distinct red edges of $f$. Also, since each red edge is crossed at most once, $e_1$ and $e_2$ cross two disjoint pairs of red edges of $f$. Denote by $c_1$ and $c'_1$ (resp. $c_2$ and $c'_2$) the two crossing points of $e_1$ (resp. $e_2$) with the boundary of $f$. Assume that $c_1, c_2, c'_2, c'_1$ occur in this clockwise order on the boundary of $f$. This implies that, while moving clockwise on the boundary of $f$, there is at least one vertex $u$ of $f$ between $c_1$ and $c_2$, and at least one vertex $v$ of $f$ between $c'_2$ and $c'_1$. Hence, we can augment $f$ with a red edge $uv$, which contradicts that we completed \textsf{Step~1}.    
        
        By properties $(i)$, $(ii)$, and $(iii)$, $f$ is traversed by exactly one green edge $e$; however this edge cannot leave on the same side two vertices of $f$ that are not consecutive on its boundary, as otherwise they would have been connected in \textsf{Step~1}. Hence $f$ is a quadrilateral and $e$ splits $f$ into two equal parts (see \cref{fi:augmentation-d}). We can then augment $f$ by adding a diagonal red edge in the quadrilateral face.  We repeat this procedure until $\Gamma_r$ contains no face $f$ such that $\deg_{\Gamma_r}(f) \geq 4$.
        %
        % is traversed by exactly one green edge $e$ in $\Gamma$, otherwise either $\Gamma$ is not min-1-planar or $f$ would have been augmented in \textsf{Step 1}; refer to \cref{fi:augmentation-c} . More precisely, let $e_1$ and $e_2$ be two green edges that traverse $f$ in $\Gamma$. By property $(ii)$, each of $e_1$ and $e_2$ crosses two distinct red edges of $f$. Also, since each red edge is crossed at most once, $e_1$ and $e_2$ cross two disjoint pairs of red edges of $f$. Denote by $c_1$ and $c'_1$ (resp. $c_2$ and $c'_2$) the two crossing points of $e_1$ (resp. $e_2$) with the boundary of $f$. Assume that $c_1, c_2, c'_2, c'_1$ occur in this clockwise order on the boundary of $f$. This implies that, while moving clockwise on the boundary of $f$, there is at least one vertex $u$ of $f$ between $c_1$ and $c_2$, and at least one vertex $v$ of $f$ between $c'_2$ and $c'_1$. Hence, we can augment $f$ with a red edge $uv$, which contradicts that we completed \textsf{Step 1}.  
        %
        % Due to properties $(i)$, $(ii)$, and $(iii)$, we can arbitrarily select a vertex $u$ of $f$ and connect it (with red edges) to each other vertex of $f$ not adjacent to $u$; see, e.g.,  \cref{fi:augmentation-d,fi:augmentation-e}. Since $f$ is traversed by one green edge, the added edges have at most one crossing, and thus the obtained drawing is min-1-planar. 
        %
    \end{compactitem}
\end{proof}

\lemaximalmintwoplanarprop*
\begin{proof} We prove the three properties separately.
\begin{compactenum}[$(a)$]
  %\item If an edge $vw$ of $G$ were to cross itself in $\Gamma$ then this crossing could be eliminated without introducing any additional crossings by replacing the curve  $c:=vw$ in $\Gamma$ with the shortest path from $v$ to $w$ that is contained in the image of $c$.  This contradicts the assumption that $\Gamma$ is crossing-minimal.
  
  \item Suppose for contradiction that $f$ does not contain an edge $uv$. Then we can add (another copy of) edge $uv$ to $\Gamma$ (and therefore to $G$) in the interior of $f$, without introducing any additional crossings or creating a $2$-real bigon. Since the resulting drawing is bundle-proper min-$k$-planar, this contradicts the hypothesis that $G$ is maximally-dense.

  %\item Suppose that two edges $e_1$ and $e_2$ cross twice at points $x$ and $y$. Since $\Gamma$ is min-$2$-planar, at least one of these edges, say $e_1$ is not involved in any additional crossings. Then we can reroute the portion of $e_2$ between $x$ and $y$ along the portion of $e_1$ between $x$ and $y$ to eliminate these two crossings without introducing any additional crossings, thereby contradicting the assumption that $\Gamma$ is crossing-minimal.
  
  \item Let $f$ be a face of $\Gamma$. Since $G$ has no self-loops and $\Gamma$ is a simple drawing, then $\deg_{\Gamma}(f) > 1$. Also, since $\Gamma$ is simple then $f$ is neither a 0-real bigon nor a 1-real bigon. Finally, since $\Gamma$ is also bundle-proper, $f$ cannot be a 2-real bigon. It follows that $\deg_{\Gamma}(f) > 2$. 

  \item Suppose $\deg_\Gamma^r(f) \geq 3$. If $\deg_\Gamma(f) \geq 4$ then there would be two non-consecutive real vertices on the boundary of $f$ that are not connected by an edge, which is impossible by $(a)$. Then $f$ is necessarily a 3-real triangle.    
  
  %\item Let $f$ be a $0$-real triangle in $\Gamma$ with edges $e_i=u_iv_i$, for $i\in\{1,2,3\}$.  Suppose for the sake of contradiction that $u_1=u_2$.  Since $\Gamma$ is min-$2$-planar and $\overline{e}_1$, $\overline{e}_2$, and $\overline{e}_3$ are pairwise crossing, at least two of $\overline{e}_1$, $\overline{e}_2$ and $\overline{e}_3$ are not involved in any additional crossings.  Therefore, at least one of $e_1$ or $e_2$, say $e_1$, is not involved in any additional crossings.  Then the crossing $x$ between $\overline{e}_1$ and $\overline{e}_2$ can be eliminated without introducing additional crossings by rerouting the portion of $\overline{e}_2$ between $v_2$ and $x$ along the portion of $\overline{e}_1$ between $v_1$ and $x$.\qedhere
\end{compactenum}
\end{proof}

\thdensitymintwo*
\begin{proof}
  We already observed at the beginning of this section that there exist min-2-planar simple graphs with $5n-10$ edges (e.g., the optimal 2-planar). It remains to prove that min-2-planar simple graphs cannot have more than $5n-10$ edges. Since any simple graph is also a bundle-proper graph, we can show that the upper bound holds more in general for multi-graphs that are bundle-proper min-2-planar. Also, since we want to find an upper bound on the number of edges, we can restrict our attention to maximally-dense bundle-proper min-2-planar graphs, and in particular to those having the minimum number of crossings. Let $(G,\Gamma)$ be any maximally-dense crossing-minimal bundle-proper min-2-planar pair, with $|V(G)|=n$.
  We show the existence of a charging function $\ch'(\cdot)$ that satisfies \textsf{(C1)} and \textsf{(C2)} for $\alpha = \frac{2}{5}$, so the result will follow from \cref{eq:modified-charge}.

  Consider the initial charging function $\ch(\cdot)$ defined in \cref{eq:initial-charge}. For each type of triangle $t$ we analyze the value of $\ch(t)$ and the deficit/excess w.r.t. $\frac{2}{5}\deg_\Gamma^r(t)$.
  \begin{itemize}
    \item If $t$ is a $0$-real triangle, $\ch(t)=-1 < 0 = \tfrac{2}{5}\deg_{\Gamma}^r(t)$, thus $t$ has a deficit of $1$. 
    \item If $t$ is a $1$-real triangle, $\ch(t)=0 < \tfrac{2}{5}=\tfrac{2}{5}\deg_{\Gamma}^r(t)$, thus $t$ has a deficit of $\frac{2}{5}$.
    \item If $t$ is a $2$-real triangle, $\ch(t)=1 > \tfrac{4}{5}=\tfrac{2}{5}\deg_{\Gamma}^r(t)$, thus $t$ has an excess of $\frac{1}{5}$.
    \item If $t$ is a $3$-real triangle, $\ch(t) = 2 > \tfrac{6}{5} = \tfrac{2}{5}\deg_{\Gamma}^r(t)$, thus $t$ has an excess of $\frac{4}{5}$.
  \end{itemize}
   Also, if $f$ is any face of $\Gamma$ with $\deg_{\Gamma}(f) \geq 4$, then $\ch(f)=2\deg_{\Gamma}^r(f)+\deg_{\Gamma}^c(f)-4 = \deg_{\Gamma}(f)-4 + \deg_{\Gamma}^r(f) \geq \deg_{\Gamma}^r(f) \geq \tfrac{2}{5}\deg_{\Gamma}^r(f)$.

  Therefore $\ch(\cdot)$ only fails to satisfy \textsf{(C1)} at $0$-real and $1$-real triangles.  We begin by setting $\ch'(f)=\ch(f)$ for each face $f$ of $\Gamma$ and we explain how to modify $\ch'(\cdot)$ in such a way that $\ch'(f) \geq \frac{2}{5}\deg_{\Gamma}^r(f)$ for each face $f \in F(\Gamma)$, thus satisfying \textsf{(C1)}, and the total charge remains the same, thus satisfying \textsf{(C2)}. 
  
  \paragraph{Fixing $0$-real triangles.}
  Let $t$ be a $0$-real triangle in $\Gamma$ with edges $e_1$, $e_2$, and $e_3$.  Refer to \cref{fi:0-real-triangles}.  The edges $\overline{e}_1$, $\overline{e}_2$ and $\overline{e}_3$ are three pairwise crossing edges of $G$. Since $\Gamma$ is a simple drawing, $\overline{e}_1$, $\overline{e}_2$ and $\overline{e}_3$ are independent edges of $G$ (i.e., their six end-vertices are all distinct). Also, since $\Gamma$ is min-2-planar, at least two of these three edges, say $\overline{e}_2$ and $\overline{e}_3$, do not cross other edges of $G$ in $\Gamma$. This implies that each of the two end-vertices of $\overline{e}_2$ shares a face with an end-vertex of $\overline{e}_3$. Hence, by      
  \cref{le:maximal-min-2-planar-prop}(\ref{edge_appears}), the four vertices of $\overline{e}_2$ and $\overline{e}_3$ form a $4$-cycle $e'\overline{e}_2e''\overline{e}_3$ in $G$ and $\Gamma$ contains a $2$-real quadrilateral $f_1$ bounded by portions of $e''$, $\overline{e}_1$, $\overline{e}_2$, $\overline{e}_3$, and a $2$-real triangle $f_2$ bounded by portions of $e'$, $\overline{e}_2$, $\overline{e}_3$. 

% \begin{figure}[tb]
%     \centering
%     \subfigure[]{\includegraphics[page=1,height=.2\textwidth]{0-real-triangles}\label{fi:0-real-triangles-a}}
%     \hfil
%     \subfigure[]{\includegraphics[page=2,height=.2\textwidth]{0-real-triangles}\label{fi:0-real-triangles-b}}
%     \hfil
%     {\includegraphics[page=3,height=.2\textwidth]{0-real-triangles}\label{fi:0-real-triangles-c}}
%     \hfil
%     \subfigure[]{\includegraphics[page=4,height=.2\textwidth]{0-real-triangles}\label{fi:0-real-triangles-d}}
%     \caption{(a) A 0-real triangle $t$. (b) A 2-real quadrilateral $f_1$ and a 2-real triangle $f_2$ neighboring $t$. (c) The initial charges. (d) The charges after a redistribution.}
%     \label{fi:0-real-triangles}
% \end{figure}
  
  The charge of $f_1$ is $\ch'(f_1)=2$, with an excess of $\frac{6}{5}$ w.r.t. $\frac{2}{5}\deg_\Gamma^r(f_1)=\frac{4}{5}$. The charge of $f_2$ is $\ch'(f_2)=1$, with an excess of $\frac{1}{5}$ w.r.t. $\frac{2}{5}\deg_\Gamma^r(f_2)=\frac{4}{5}$. We reduce $\ch'(f_1)$ by $\frac{4}{5}$, reduce $\ch'(f_2)$ by $\frac{1}{5}$, and increase $\ch'(t)$ by $1$.
  After that, the total charge is unchanged and all the three faces $t$, $f_1$, and $f_2$ satisfy \textsf{(C1)}. Namely, $\ch'(t)=0$ (it has no deficit/excess), $\ch'(f_1)=\frac{6}{5}$ (it has an excess of $\frac{2}{5}$), and $\ch'(f_2)=\frac{4}{5}$ (it has no deficit/excess).   
  In the remainder of the proof, since we need a way to keep track of the $2$-real triangles and $2$-real quadrilaterals whose charge has been modified as described above, we call each of the faces $f_1$ and $f_2$ a \emph{$0$-real triangle-neighboring face}.  Each $0$-real triangle-neighboring face that is a $2$-real triangle (as $f_2$) shares its unique $0$-real vertex with a $0$-real triangle; each $0$-real triangle-neighboring face that is a $2$-real quadrilateral (as $f_1$) shares its unique $0$-real edge with a $0$-real triangle.
  
  \paragraph{Fixing $1$-real triangles.}
  Let $t$ be a $1$-real triangle, with real-vertex $v_1$ and crossing-vertices $v_2$ and $v_3$.  Refer to \cref{fi:1-real-triangles} for an illustration. Let $e_0=v_2v_3$ be the $0$-real edge of $t$, and let $f_1$ be the face of $\Gamma$ that shares $e_0$ with $t$. If $f_1$ is a $0$-real quadrilateral, denote by $e_1$ the $0$-real edge of $f_1$ not adjacent to $e_0$, and by $f_2$ the face of $\Gamma$ that shares $e_1$ with $f_1$. If $f_2$ is a $0$-real quadrilateral, denote by $e_2$ the $0$-real edge of $f_2$ not adjacent to $e_1$, and by $f_3$ the face of $\Gamma$ that shares $e_2$ with $f_2$. We continue in this way until we encounter a face $f_p$ $(p \geq 1)$ that is not a $0$-real quadrilateral. This procedure determines a sequence of faces $f_0, f_1, f_2, \dots f_p$, and a sequence of $0$-real edges $e_0, e_1, \dots, e_{p-1}$ such that $f_0=t$, $f_i$ is a $0$-real quadrilateral for each $i \in \{1, \dots, p-1\}$, $f_p$ is not a $0$-real quadrilateral, and the faces $f_i$ and $f_{i-1}$ share edge $e_{i-1}$ $(i \in \{1, \dots, p\})$. 

%   \begin{figure}[tb]
%     \begin{center}
%     \includegraphics[page=1,height=.25\textwidth]{1-real-triangles}
%     \end{center}
%     \caption{The demand path of a $1$-real triangle $t$, ending at a face $f_p$.}
%     \label{fi:1-real-triangles}
%   \end{figure}
  
  Note that $\deg_\Gamma(f_p) \geq 4$. Namely, let $e=v_1v_2$ and $e'=v_1v_3$, and let $\overline{e}=v_1u$ and $\overline{e'}=v_1w$ be the edges of $G$ that contain $e$ and $e'$. Since $f_p$ has at least two crossing-vertices, if $f_p$ were a triangle then it would be either a $0$-real triangle or a $1$-real triangle.
  If $f_p$ were a $0$-real triangle then $\overline{e}$ and $\overline{e'}$ would cross in $\Gamma$, which is impossible as $\overline{e}$ and $\overline{e'}$ are adjacent edges and $\Gamma$ is a simple drawing. If $f_p$ were a $1$-real triangle then $u=w$, i.e., $\overline{e}$ and $\overline{e'}$ would be parallel edges both involved in a crossing, which is impossible as $\Gamma$ is bundle-proper. 
  
  Therefore, $\deg_\Gamma^r(f_p) \geq 4$ and, as already observed at the beginning of this proof, $\ch'(f_p) \geq \frac{2}{5} \deg_\Gamma^r(f_p)$. Also, the charge excess of $f_p$ is larger than $\frac{2}{5}$. Namely, the charge excess of $f_p$ is $x = 2\deg_\Gamma^r(f_p) + \deg_\Gamma^c(f_p)-4 - \frac{2}{5}\deg_\Gamma^r(f_p) = \deg_\Gamma(f)  + \frac{3}{5}\deg_\Gamma^r(f) - 4$.
  If $f_p$ has no real-vertices, then it must have at least five crossing-vertices (because $f_p$ is not a $0$-real quadrilateral), which implies $x \geq 1 > \frac{2}{5}$. If $f_p$ has at least one real-vertex then $x \geq \frac{3}{5} > \frac{2}{5}$.
  
  Hence, since the charge excess of $f_p$ is larger than $\frac{2}{5}$, the idea is to fill the $\frac{2}{5}$ charge deficit of $t$ by moving an equivalent amount of charge from $f_p$ to $t$. We say that  
  $t$ \emph{demands from $f_p$ through edge $e_{p-1}$} a charge amount of $\tfrac{2}{5}$. We call $f_0,\ldots,f_p$ (which is a path in the dual of $\Gamma$) the \emph{demand path} for $t$. Therefore, for each $1$-real triangle $t$ of $\Gamma$ whose demand path ends at a face $f=f_p$, we decrease $\ch'(f)$ by $\tfrac{2}{5}$ and increase $\ch'(t)$ from $0$ to $\tfrac{2}{5}$. Note that $f$ cannot be a $0$-real triangle-neighboring face. Indeed, $f$ is not a triangle, and if $f$ is a $2$-real quadrilateral then its $0$-real edge is shared either with a $0$-real quadrilateral or directly with the $1$-real triangle $t$. It follows that the set of faces whose charge is affected by fixing $1$-real triangles does not intersect with the set of faces whose charge is affected by fixing $0$-real triangles.
  
  Due to the considerations above, after we have fixed all 1-real triangles, we may have problems only if multiple $1$-real triangles demanded from the same face $f$. In this case, $f$ might no longer satisfy \textsf{(C1)}. In the remainder of the proof, we analyze which types of faces may be in this situation and, if so, we prove how to fix their charge.

  \paragraph{Fixing faces that received multiple demands from 1-real triangles.}
  Let $f$ be a face of $\Gamma$ of degree larger than three that received multiple demands from 1-real triangles. This is possible only if $f$ has more than one 0-real edge, hence we can exclude that $f$ is a $2$-real quadrilateral. Note that, by \cref{le:maximal-min-2-planar-prop}(\ref{3_real_triangle}), each face of $\Gamma$ contains at most three real-vertices. If $\deg_\Gamma(f) \geq 7$ then $f$ still satisfies \textsf{(C1)} even if it received a demand through each of its $\deg_\Gamma(f)$ edges when fixing $1$-real triangles. Indeed, in the worst case, the new charge of $f$ is $\ch'(f) = \deg_\Gamma(f) + \deg_\Gamma^r(f) - 4 - \frac{2}{5}\deg_\Gamma(f) = \frac{3}{5}\deg_\Gamma(f) + \deg_\Gamma^r(f) - 4 \geq \frac{2}{5}\deg_\Gamma^r(f)$ (because $\deg_\Gamma(f) \geq 7$). 
  The same happens if $\deg_\Gamma(f) \geq 5$ and $\deg_\Gamma^r(f) \geq 1$ (i.e., $f$ has at least one real-vertex). Indeed, in this case, the number of $0$-real edges of $f$ is at most $\deg_\Gamma^c(f) -1 = \deg_\Gamma(f) - \deg_\Gamma^r -1$, so $f$ received at most this number of demands from $1$-real triangles. Hence, in the worst case, the new charge of $f$ is $\ch'(f) = \deg_\Gamma(f)+\deg_\Gamma^r(f)-4-\frac{2}{5}(\deg_\Gamma(f)-\deg_\Gamma^r(f)-1) \geq \frac{3}{5}\deg_\Gamma(f) + \frac{7}{5}\deg_\Gamma^r(f) - \frac{18}{5} \geq \frac{7}{5}\deg_\Gamma^r(f) -\frac{3}{5} \geq \frac{2}{5}\deg_\Gamma^r(f)+\deg_\Gamma^r(f)-\frac{3}{5} \geq \frac{2}{5}\deg_\Gamma^r(f)$. 
  It follows that the only faces that may have received multiple demands from 1-real triangles and that (after we have fixed all 1-real triangles) no longer satisfy \textsf{(C1)} are the $1$-real quadrilaterals, the $0$-real pentagons, and the $0$-real hexagons. Each face $f$ of one of these types has at least two adjacent $0$-real edges. If $f$ no longer satisfies \textsf{(C1)}, we show how to find extra charges that can be moved from some suitable faces with charge excess towards $f$, so to compensate the charge deficit of $f$. To this aim, we first prove the following claim; refer to \cref{fi:leaks}.
  
  \begin{claim}
    Let $f$ be a face of $\Gamma$ and let $e_1, e_2 , e_3, e_4$ be consecutive edges on the boundary of $f$ for which a demand is made through both $e_2$ and $e_3$. Let $t_1$ be the $1$-real triangle that demanded from $f$ through $e_2$ and let $v_1 = \overline{e}_1\cap \overline{e}_3$ be the real-vertex of $t_1$. Analogously, let $t_2$ be the $1$-real triangle that demanded from $f$ through $e_3$ and let $v_2=\overline{e}_2\cap \overline{e}_4$ be the real-vertex of $t_2$.  Then there is a curve $C$ that begins in $f$, leaves $f$ through the crossing-vertex common to $e_2$ and $e_3$, passes through a sequence of zero or more $1$-real triangles and $1$-real edges, and ends in a face $f^*$ that is either a $2$-real triangle containing $v_1$ and $v_2$ or a $2$-real quadrilateral containing only one of $v_1$ and $v_2$. 
  \end{claim}

  \begin{nestedproof}
    Observe that the closed region $\Delta_{123}$ bounded by $\overline{e}_1$, $\overline{e}_2$, $\overline{e}_3$ does not contain any vertex of $G$ other than $v_1$ since if it did, then $t_1$ would be making a demand from some face other than $f$. Similarly, the closed region $\Delta_{234}$ bounded by $\overline{e}_2$, $\overline{e}_3$, $\overline{e}_4$ contains no vertices of $G$.

\begin{figure}[tb]
    \centering
    \subfigure[]{\includegraphics[page=2,width=.38\textwidth]{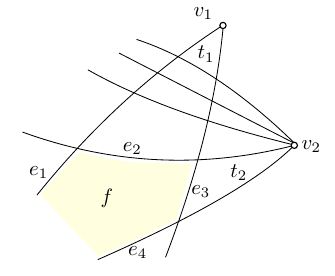}\label{fi:leaks-a}}
    \hfil
    \subfigure[]{\includegraphics[page=3,width=.38\textwidth]{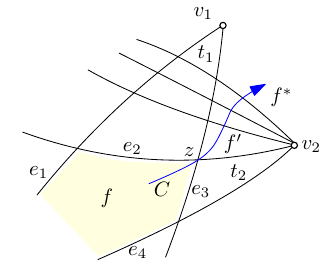}\label{fi:leaks-b}}
    \hfil
    \subfigure[]{\includegraphics[page=4,width=.38\textwidth]{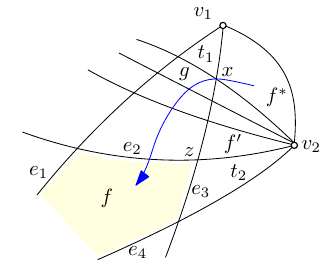}\label{fi:leaks-c}}
    \hfil
    \subfigure[]{\includegraphics[page=5,width=.38\textwidth]{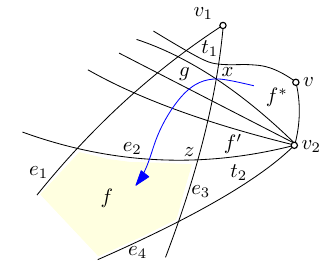}\label{fi:leaks-d}}
    \caption{
    A supporting face $f^*$ for a face $f$ that receives charge demands through two consecutive $0$-real edges of its boundary ($e_2$ and $e_3$ in the figure).}
    \label{fi:leaks}
\end{figure}

      We can construct a curve $C$ that begins in $f$, passes through the crossing-vertex $z$ of $\Gamma$ common to $e_2$ and $e_3$, and then enters the face $f'$ opposite $f$ at $z$.
      From the interior of $f'$ the curve $C$ then crosses a sequence of zero or more $1$-real edges incident to $v_2$ and passes through zero or more $1$-real triangles that contain $v_2$ until reaching some face $f^*$ that contains $v_2$ and is not a $1$-real triangle. One of the following must occur:
      \begin{compactenum}
      
        \item The face $f^*$ contains $v_1$ (see \cref{fi:leaks-c}).  In this case \cref{le:maximal-min-2-planar-prop}(\ref{edge_appears}) implies that $f^*$ also contains the $2$-real edge $v_1v_2$.  The crossing-minimality of $\Gamma$ implies that $f^*$ is a $2$-real triangle that contains $v_1$ and $v_2$.  (Otherwise $\overline{e}_3$ has more than one crossing on the boundary of $f^*$ and could be rerouted to avoid all but one of these crossings.) %\todo{A: Non ho capito la cosa tra parentesi. W: In fact, I think this situation cannot happen without affecting the definition of $t_1$, but I would keep it for the moment to be on the right side.}
        
        \item The face $f^*$ has degree larger than three (see \cref{fi:leaks-d}). Then $f^*$ contains an edge $e$ of $\Gamma$ that is not incident to $v_1$ and $\overline{e}$ crosses $\overline{e}_3$. No endpoint of $\overline{e}$ is in $\Delta_{123}$, $\overline{e}$ does not cross $e_2$, and since $\Gamma$ is simple, $\overline{e}$ has only one crossing with $\overline{e}_3$, so $\overline{e}$ must cross $\overline{e}_1$.  Since $e_3$ already crosses $e_2$ and $e_4$, this implies that $\overline{e}$ has no additional crossings.  Therefore one end-vertex $v$ of $\overline{e}$ belongs to $f^*$. By \cref{le:maximal-min-2-planar-prop}(\ref{edge_appears}), the edge $v_2v$ is on $f^*$.  The crossing minimality of $\Gamma$ then implies that $f^*$ is a $2$-real quadrilateral that contains $v_2$.
      \end{compactenum}
      This completes the proof of the claim.
    \end{nestedproof}
    
  Let $f$ be a $1$-real quadrilateral, a $0$-real pentagon, or $0$-real hexagon with edges $e_1,e_2,e_3,e_4$ that satisfy the conditions of the claim and let $f^*$ be the face whose existence is established by the claim. In each such case, we move a charge of $\tfrac{1}{5}$ from $f^*$ to $f$.
  Based on the claim, there are two cases to consider:

  \begin{compactenum}
    \item The face $f^*$ is a $2$-real triangle that contains $v_1$ and $v_2$ (see \cref{fi:leaks-c}).  Let $x$ be the crossing-vertex of $f^*$.  Then the face $g$ that shares $x$ with $f^*$ but has no edge in common with $f^*$ is either a $0$-real quadrilateral or it coincides with $f$.  In particular, $g$ is not a $0$-real triangle, which implies that $\ch'(f^*)$ was not modified when fixing $0$-real triangles.  Therefore $\ch'(f^*)=1$ immediately after fixing $0$-real triangles.  The charge on $2$-real triangles is never modified when fixing $1$-real triangles, hence $\ch'(f^*)=1$ even after fixing all $1$-triangles.
    Since we reduce the charge of $f^*$ by $\tfrac{1}{5}$ and increase the charge of $f$ by $\tfrac{1}{5}$, we can think of this charge travelling along the suffix of the demand path $t_1\leadsto f$ that begins at $g$; we also say that the charge \emph{leaks out of $f^*$ through $x$}.

    We show that charge leaks out of $f^*$ through $x$ at most once. This is obviously the case if $f^*$ and $f$ share the vertex $x=z$ (i.e., $g=f$).  The only other possibility is that the charge leaks out of $f^*$ into the $0$-real quadrilateral $g$ that is part of another demand path $t \leadsto f''$, with $t\neq t_1$.  Let $e$ and $e'$ be the two edges of $g$ other than $e_1$ and $e_3$.  Then $t$ is the $1$-real triangle that contains $v_2$ and whose $1$-real edges are portions of $\overline{e}$ and $\overline{e}'$.  Each of $\overline{e}$ and $\overline{e}'$ crosses $e_1$ and $e_3$.  Possibly $e_2\in\{e,e'\}$ but we can assume without loss of generality that $e\neq e_2$.
    Therefore the edge $\overline{e}_3$ crosses $\overline{e}_2$, $\overline{e}_4$, and $\overline{e}$, for a total of at least $3$ crossings.
    Hence, neither $e$ nor $e'$ is involved in any additional crossings, which implies that the face next to $g$ on the demand path $t \leadsto f''$ contains end-vertices of $\overline{e}$ and $\overline{e}'$, i.e., this face coincides with $f''$.  Hence, since $\Gamma$ is simple (which excludes that $\overline{e}$ and $\overline{e}'$ cross), $\deg_\Gamma^r(f'') \geq 2$. It follows that $f''$ is neither a $1$-real quadrilateral, nor a $0$-real pentagon, nor a $0$-real hexagon.  Since the charge that leaks out of $f^*$ through $x$ is always left at a $1$-real quadrilateral, or at a $0$-real pentagon, or at a $0$-real hexagon, we conclude that charge leaks out of $f^*$ at $x$ at most once.

    \item The face $f^*$ is a $2$-real quadrilateral (see \cref{fi:leaks-d}).  In this case, $f^*$ has only one $0$-real edge, shared with a $0$-real quadrilateral.  Again, this implies that $\ch'(f^*)$ was not modified when fixing of $0$-real triangles. Hence, immediately after fixing $0$-real triangles we have $\ch'(f^*)=2$.
    Since we reduce $\ch'(f^*)$ by $\tfrac{1}{5}$ and increase $\ch'(f)$ by $\tfrac{1}{5}$, we can again think of this as a charge of $\tfrac{1}{5}$ leaking from $f^*$ through a $0$-real vertex $x$ of $f^*$ and then travelling down a suffix of the demand path $t_1 \leadsto f$.  With the same reasoning as above, this can happen at most once for each of the two $0$-real vertices of $f^*$.  Therefore, the total charge that leaves through these two vertices of $f^*$ is at most $\frac{2}{5}$.
  \end{compactenum}

 \begin{figure}
    \centering
    \includegraphics[page=6,height=.38\textwidth]{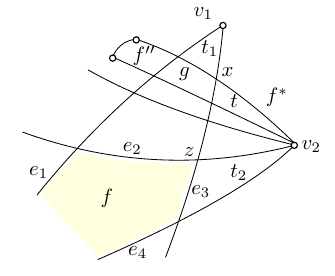}
    \caption{A supporting face $f^*$ supports at most one face $f$.}
    \label{fi:leaks-further}
  \end{figure}

  \medskip
  To summarize the discussion above, each face $f^*$ can give a charge of $\frac{1}{5}$ to at most another face $f$, and after that it still satisfies \textsf{(C1)}. In the following we call $f^*$ a \emph{supporting face}.
  To complete the proof, we have to show that if $f$ is either a $1$-real quadrilateral, or a $0$-real pentagon, or a $0$-real hexagon, and if $f$ received multiple demands from $1$-real triangles, then $f$ always finds a suitable number of supporting faces to satisfy \textsf{(C1)}. We analyze separately the three different types of categories for $f$, and assume that $f$ received more than one demand from some $1$-real triangles.   

  \begin{itemize}
  \item If $f$ is a $1$-real quadrilateral, then it received exactly two demands of $\frac{2}{5}$, through its two consecutive $0$-real edges. In this case, a charge of $\frac{1}{5}$ leaks into $f$ from one supporting face. Hence we have $\ch'(f) \geq 1 + \frac{1}{5} - 2\cdot\frac{2}{5} = \frac{2}{5} = \frac{2}{5} \deg_\Gamma^r(f)$, that is $f$ satisfies \textsf{(C1)}.

  \item If $f$ is a $0$-real pentagon, the following cases are possible:
    \begin{itemize}
      \item $f$ received exactly two demands. We have $\ch'(f) = 1 - 2\cdot\frac{2}{5} = \frac{1}{5} > 0 = \frac{2}{5}\deg_\Gamma^r(f)$, thus $f$ satisfies \textsf{(C1)} without needing supporting faces.
      \item $f$ received exactly three demands. Two of these demands necessarily occur through two consecutive edges of $f$, so a charge of at least $\tfrac{1}{5}$ leaks into $f$ from a supporting face. Therefore $\ch'(f) = 1 + \frac{1}{5} - 3\cdot\frac{2}{5}=0=\frac{2}{5}\deg_\Gamma^r(f)$, that is $f$ satisfies \textsf{(C1)}.
      \item $f$ received exactly four demands. There are three pairs of consecutive edges of $f$ at which the demands occur, so a total charge of $\tfrac{3}{5}$ leaks into $f$ from three supporting faces. Therefore $\ch'(f) = 1 +\frac{3}{5} - 4\cdot\frac{2}{5}=0=\frac{2}{5}\deg_\Gamma^r(f)$, that is $f$ satisfies \textsf{(C1)}.
      \item $f$ received exactly five demands. There are five pairs of consecutive edges at which the demands occur, so a total charge of $\frac{5}{5}=1$ leaks into $f$ from five supporting faces. Therefore $\ch'(f) = 1+1-5\cdot\tfrac{2}{5}=0=\frac{2}{5}\deg_\Gamma^r(f)$, that is $f$ satisfies \textsf{(C1)}.
    \end{itemize}
    \item If $f$ is a $0$-real hexagon, we have two cases. If $f$ received at most five demands, then $\ch'(f) \geq 2 - 5\cdot\frac{2}{5}=0=\tfrac{2}{5}\deg_\Gamma^r(f)$. Otherwise, a total charge of at least $\frac{6}{5}$ leaks into $f$ from six supporting faces, thus we have $\ch'(f) = 2+\frac{6}{5}-6\cdot\frac{2}{5} = \frac{4}{5}>0=\frac{2}{5}\deg_\Gamma^r(f)$. Hence, in both cases $f$ satisfies \textsf{(C1)}. 
\end{itemize}
      
In conclusion, at the end of the discharging process, the new function $\ch'(\cdot)$ satisfies \textsf{(C1)} for all faces of $\Gamma$, and the total charge is the same as the initial total charge, that is, $\ch(\cdot)$ satisfies \textsf{(C2)}. This completes the proof.
\end{proof}

\thheavymintwo*
\begin{proof}
    We apply a simplified version of the charging technique as in \cref{th:density-min2}, to achieve the number of heavy edges. Here is a sketch of the proof. From $G$ with a corresponding drawing $\Gamma$, we derive graph $G^*$ with drawing $\Gamma^*$ with the same number of vertices and the same heavy edges crossing the same edges as in $\Gamma$, and which induces a 1-planar triangulation when the heavy edges being removed.  
    %Depending on the number of adjacent vertices, we distinguish 0-real triangles, 1-real triangles, 2-real triangles and 3-real triangles. Then we consider the faces which are traversed by heavy edges.
    We observe that %0-real triangles and 1-real triangles are not traversed at all by heavy edges, therefore they will be ignored in the following. 
    this triangulation contains only 2-real and 3-real triangles, to which we give, as before, 1 unit or 2 units of charge, respectively, for a total of $4n - 8$. 
    We distribute the charge of those faces to the traversing segments of the heavy edges such that each heavy edge receives at least $\frac{10}{3}$ charge. From this, we can conclude an upper bound of $\frac{6}{5}(n-2)$ heavy edges.

More formally, let $G$ be any min-2-planar graph with a corresponding drawing $\Gamma$ with the maximum number of heavy edges.
Note that all heavy edges in $\Gamma$ have at least three crossings with light edges.
From $\Gamma$, we derive a graph $G^-$ and drawing $\Gamma^-$ by removing all the heavy edges in $\Gamma$. Then we remove all light edges that have two crossings in $\Gamma^-$. Note that the edges that have been originally crossed the heavy edges will not be removed during this phase. The new drawing $\Gamma'$ is 1-planar. %Observe that inserting heavy edges would not violate min-1-planarity. 
All the faces of $\Gamma'$ describe cyclic sequences of real vertices and crossing-vertices, and in each sequence we do not have two subsequent crossing-vertices, as this would mean two crossings on a light edge. 

    \begin{claim}
    $\Gamma'$ can be extended by additional light edges to a $\Gamma^+$ such that all faces are triangles and
$\Gamma^*$ achieved by inserting the heavy edges is still min-2-planar.         
    \end{claim}
    \begin{nestedproof}
        This insertion of extra edges for triangulation can be done using similar techniques as in \cref{le:density-min1-support}. If $u$ and $v$ are vertices of two distinct 1-real edges with a crossing $c$, then we can add an edge between $u$ and $v$ having at most two crossings with heavy edges $e$ and $e'$ (see \cref{fi:augmentation-min2-a}). So we assure that faces $f$ with $\deg_{\Gamma^+}(f) \ge 4$ have only real vertices. We can now reduce inductively $\deg_{\Gamma^+}(f)$. If $f$ is not traversed by a heavy edge, we can triangulate it instantly. If $f$ is traversed by a heavy edge $e$ which is incident to a vertex $u$ of $f$ and crosses the boundary of $f$ through the edge $vv'$, we can add an edge $uv$ which has at most one crossing with a heavy edge $e'$ (see \cref{fi:augmentation-min2-b}). Note that instead of $uv$ we can choose $uv'$, if $u$ and $v$ are already adjacent. Assume now that no heavy edge $e$ traversing $f$ is incident to any vertex of $f$. Let $u,u',v,v'$ be vertices of $f$ and $uu',vv'$ the edges which are crossed by $e$. Then we can add the edge $uv$ as it has at most two crossings with heavy edges $e', e''$ (see \cref{fi:augmentation-min2-c}). Note again that instead of $uv$ we may choose $u'v'$, if $u$ and $v$ are already adjacent. If $u'$ and $v'$ are also adjacent, this implies $\deg_{\Gamma^+}(f) = 4$. In this case there are at most four heavy edges traversing $f$ and all their combinations allow adding an edge between two vertices of $f$.
    \begin{figure}[b]
	\centering
    \subfigure[]{\includegraphics[width=.185\textwidth]{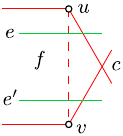}\label{fi:augmentation-min2-a}}
    \hfil
    \subfigure[]{\includegraphics[width=.27\textwidth]{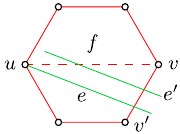}\label{fi:augmentation-min2-b}}
    \hfil
    \subfigure[]{\includegraphics[width=.27\textwidth]{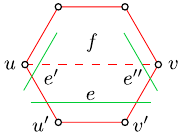}\label{fi:augmentation-min2-c}}
    \caption{
    The triangulation in the proof of \cref{th:heavy-min2}.}
    \label{fi:augmentation-min-2}
    \end{figure}
    \end{nestedproof}
    
    Since each edge in $\Gamma^+$ has at most one crossing, there are no 0-real triangles nor 1-real triangles in $\Gamma^+$. See \cref{fi:traversed-triangles-min-2-a} for an illustration of a heavy edge traversing such a graph. For this reason we show properties of 2-real and 3-real triangles.

    \begin{claim}
        Let $f$ be a 2-real triangle in $\Gamma^+$. Then the following holds:
        $(i)$ If $f$ is a start face for a heavy edge $e$, it is the only heavy edge traversing $f$.
        $(ii)$ At most two heavy edges traverse $f$.
        $(iii)$ If two heavy edges traverse $f$, then the face $f'$ that shares only a crossing-vertex with $f$ is not traversed by a heavy edge.
    \end{claim}

    \begin{nestedproof}
        $(i)$ This is clear as one edge of $f$ has already 2 crossings and a heavy edge through the other edges of $f$ would imply two crossing heavy edges. 
        $(ii)$ Because of $(i)$ we know that this is true, if a heavy edge starts at $f$. Otherwise three or more traversing heavy edges would imply at least six extra crossings with the border of $f$ and this contradicts that $\Gamma^*$ is min-2-planar. $(iii)$
        Both 1-real edges of $f$ are also edges of $f'$. Because of $(i)$ both heavy edges do not start at $f$ and therefore the border of $f$ is crossed four times. So the 1-real edges of $f$ already have each two crossings and no heavy edge can traverse $f'$ (see \cref{fi:traversed-triangles-min-2-b}).
    \end{nestedproof}

\begin{figure}[hb]
	\centering
    \subfigure[]{\includegraphics[width=.50\textwidth]{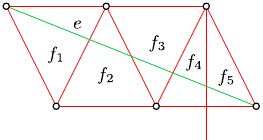}\label{fi:traversed-triangles-min-2-a}}
    \hfil
    \subfigure[]{\includegraphics[width=.28\textwidth]{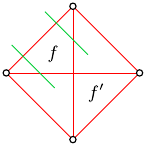}\label{fi:traversed-triangles-min-2-b}}
    \caption{
    Heavy edges crossing 2-real and 3-real triangles. $(a)$ $f_1$ and $f_5$ are start faces for the the edge $e$ and $f_2, f_3, f_4$ intermediate faces. $(b)$ If $f$ is traversed by two heavy edges, then $f'$ is traversed by no heavy edges.}
    \label{fi:traversed-triangles-min-2}
    \end{figure}
    \begin{claim}
        Let $f$ be a 3-real triangle in $\Gamma^+$. Then the following holds:
        $(i)$ If $f$ is a start face for a heavy edge $e$, there is at most one other heavy edge traversing $f$.  
        $(ii)$ At most three heavy edges traverse $f$.
    \end{claim}
    \begin{nestedproof}
        $(i)$ Let $e'$ be the edge of $f$ which is crossed by $e$ and observe that $e$ divides the triangle in two parts, so that all other heavy edges must also cross $e'$. So the maximal number of crossing of $e'$ limits the number of heavy edges traversing $f$ to two.
        $(ii)$ Note that the border of $f$ can have at most six crossings. If no traversing heavy edge starts at $f$, this implies at most three traversing heavy edges. Otherwise we know because of $(i)$ that the result is true. 
    \end{nestedproof}

Now we apply the charging technique from before. In total, we distribute $4n-8$ charge such that the 2-real triangles receive 1 charge and the 3-real triangles receive 2 charge. We also move 1 charge from a 2-real triangle $f'$ to the 2-real triangle $f$ that shares only the crossing-vertex with $f'$, if $f'$ is not traversed by a heavy edge. Then we assign the charge to the corresponding parts of the heavy edges which traverse those triangles. Because of the above statements we know that we can assign to each start segment 1 charge and to each intermediate segment $\frac{2}{3}$ without having a face with negative charge. In total, each heavy edge receives at least $\frac{10}{3}$ charge, as it has two start segments and at least two intermediate segments.
Hence overall, there cannot be more than $\frac{4 \cdot 3}{10}(n-2) = \frac{6}{5}(n-2)$ heavy edges.

\medbreak
For the lower bound, we refer to the middle part of \cref{fi:min-k-heavy-edge-lower-bound}.
\begin{figure}[htb]
    \includegraphics[width=.95\textwidth]{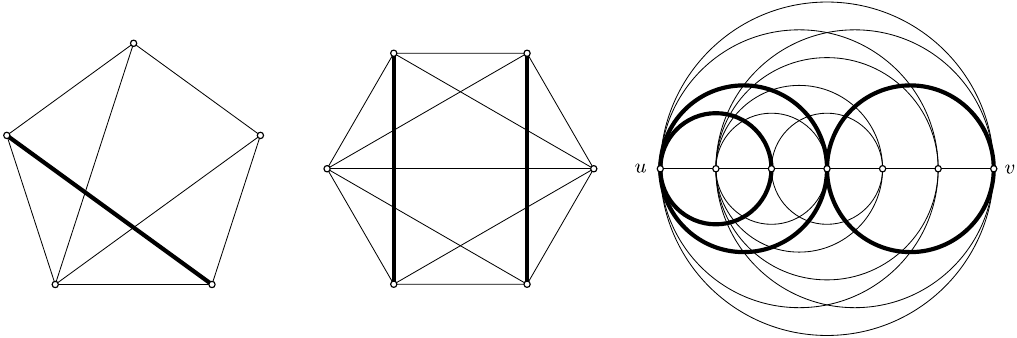}
    \caption{Lower bound examples for the number of heavy edges for min-1-planar, min-2-planar and (non-simple) min-3-planar drawings respectively. For min-1-planar drawings, we show one pentagon of a planar pentagonalization with one heavy edge for each pentagon. The bound follows from Euler's formula.  For min-2-planar drawings, we use an almost complete hexagonalization with the indicated hexagonal tile as shown. 
    The hexagons contain 2 heavy edges. The bound follows from Euler's formula since there are approximately $\frac {2n-4} 4$ hexagons. The rightmost subfigure follows the construction of \cref{fi:8-gons}. Note that the pattern can be repeated arbitrarily often with the poles $u$ and $v$.}
    \label{fi:min-k-heavy-edge-lower-bound}
\end{figure}
\end{proof}

\thdensityminthree* 
\begin{proof}
    As in the proof of \cref{th:density-min2} 
    we can assume to restrict to maximally-dense crossings-minimal bundle-proper min-3-planar pairs $(G,\Gamma)$ with $|V(G)|=n$ vertices, and we will use a discharging technique to prove the statement.
    Our discharging function $\ch'(\cdot)$ is similar to that in \cite{ackerman:on}, where the same bound for $4$-planar graphs was proven, but the details differ, as here edges with more than four crossings can exist. 
    We can assume that $\Gamma$ is 2-connected and hence that the boundary of each face $f$ of $\Gamma$ is a simple cycle. Indeed, Ackerman proved that if $\Gamma$ is not 2-connected it always has no more than $6n-12$ edges \cite[Proposition~2.1]{ackerman:on}. Although Ackerman concentrates on 4-planar graphs, he does not use this hypothesis to show this fact, thus his argument works also~in~our~case. 
    %We first show that the theorem holds for drawings that are not 2-connected, then we will focus on 2-connected graphs. 
    
    We introduce the discharging steps of $\ch'(\cdot)$ and show that all faces satisfy \textsf{(C1)} and \textsf{(C2)} for $\alpha = \frac{1}{3}$, which implies with \cref{eq:modified-charge} the desired bound on the edge density. This is relatively difficult to see for 0-real pentagons, so the major part of the proof is reserved for these. For that we distinguish different cases depending on the structure of the graph near a 0-real pentagon.

    Before we can write down the discharging function, we introduce some definitions. Let $x$ be a crossing-vertex in a drawing $\Gamma$. We call the faces $f$ and $f'$ \textit{vertex-neighbors}, if both their boundaries contain $x$ but not a common edge $e \in E(\Gamma)$. 
    Recall the definition of \textit{demand path} in the proof of \cref{th:density-min2} for 1-real triangles. 
    We naturally extend the definition of demand path for arbitrary faces $f_0$ through each of their 0-real edges. If $f$ is the end of a demand path for $f_0$, then we say that $f$ and $f_0$ are \textit{demand-path-neighbors} (even if no demand is made). Note that in \cref{sse:edge-density-min-2} $\deg_\Gamma(f) \ge 4$ was already shown for the case that $f_0$ is a 1-real triangle and this holds also for 0-real triangles, still using the fact that $\Gamma$ is a simple drawing and two edges cannot cross more than once.
    
    % Commented by Walter
    % ====================
    % Recall the definition of \textit{demand path} in the proof of \cref{th:density-min2} and note that we can generalize it so that all faces can make a demand through each of their 0-real edges. If there is a demand path that starts at a face $f$ and ends at a face $f'$, we call $f$ and $f'$ \textit{demand-path-neighbors} (even if no demand is made). Note that in \cref{sse:edge-density-min-2} $\deg_\Gamma(f') \ge 4$ was already shown for the case that $f$ is a 1-real triangle and this holds also for 0-real triangles.\\
    % ===================================
    We assign to every face $f \in F(\Gamma)$ the initial charge $\ch'_0(f) = \ch(f)$ as defined in \cref{eq:initial-charge} and modify the charges by the following steps $i \in \{1,2,3,4\}$ to get $\ch'_i(f)$. The final charge is $\ch'(f) = \ch'_4(f)$. The idea is as in \cref{sse:edge-density-min-2} to fix first the charge of the 0-real and then the 1-real triangles. The last step fixes 0-real pentagons, which contribute to multiple triangles in the first step.
	%F(\Gamma) \to \mathbb{R}$ such that $\sum_{f\in F(\Gamma)} \ch(f) \ge \sum_{f\in F(\Gamma)} \ch'(f)$ and $\ch'(f) \ge \frac{1}{3}v(f)$. Showing the existence of $\ch'(f)$ proofs the above Lemma by applying \textcolor{red}{Eq. (3)} with $\alpha = \frac{1}{3}$.\\
	%Only for $0$-triangles and $1$-triangles the initial charging functions fails. So in the first three steps we charge them in a way, that these are satisfied. By doing this eventually some $0$-pentagons loose to much charge, so the last step fixes this.
    \smallskip
	\begin{compactitem}
	    \item[\textsf{Step 1.}] Every 0-real triangle receives $\frac{1}{3}$ from each of its demand-path-neighbors.
		\item[\textsf{Step 2.}] If $f$ is a face with positive charge that is not a 1-real quadrilateral, then $f$ gives $\frac{1}{6}$ to each 1-real triangle that shares a 1-real edge with $f$. However, if $f$ is a 2-real triangle that shares only one of its two 1-real edges with a 1-real triangle $t$, then $f$ gives $\frac{1}{3}$ to $t$.
        \item[\textsf{Step 3.}] Let $t$ be a 1-real triangle with $\ch'_2(t) < \frac{1}{3}$ charge. Then $t$ receives $\frac{1}{3} - \ch'_2(t)$ from its demand-path-neighbor.
	\item[\textsf{Step 4.}] Every face $f$ distributes its (positive) excess $\ch'_3(f) - \frac{1}{3}\deg_\Gamma^r(f)$ equally over all 0-real pentagons that are vertex-neighbors of $f$.
	\end{compactitem}
    
    \smallskip
    Since charge is only moved, \textsf{(C2)} holds for $\ch'(f)$. For the four steps we have:
    
	\begin{restatable}{proposition}{prominthreeprop} \label{pro:min3-chargingprop}
        In any maximally-dense crossing-minimal bundle-proper min-3-planar pair $(G, \Gamma )$ the following holds for the above charging function $\ch'(f)$:
		\begin{itemize}
            \item[(a)] In \textsf{Step 1} and \textsf{Step 3} charge gets only moved through 0-real edges.
            \item[(b)] In \textsf{Step 2} charge gets only moved through 1-real edges.
            \item[(c)] Charge gets never moved through 2-real edges.
            \item[(d)] In \textsf{Step 3} each 1-real triangle receives most $\frac{1}{3}$ charge.
			%\item[(c)] $2$-triangles are only contributing in step 2 or 4.
			\item[(e)] Let $f$ be a 1-real triangle and $e$ a 1-real edge of the boundary of $f$. If $\overline{e}$ has more than three crossings, then $f$ receives at least $\frac{1}{6}$ charge in \textsf{Step 2}.
			%\item[(d)] Let $f$ be a face with $v(f)$ real vertices and $\deg_\Gamma(f) \ge 5$ the number of real vertices and crossings. Let further be $p$ the number of wedge-neighbors, to which $f$ distribute charge in step 1 or 3, and $q$ the number of $0$-pentagons, which are vertex-neighbors of $f$. Then $f$ contributes in step 4 to each $0$-pentagon, which is one of its vertex-neighbors, at least
			%\begin{align*}
			%	\frac{\deg_\Gamma(f) + \frac{2}{3}v(f)-4-\frac{1}{3}-\frac{p}{3}}{q}
			%\end{align*}
			%charge.
		\end{itemize}
	\end{restatable}
	\begin{nestedproof}
		$(a)$ follows by the definition of a demand-path-neighbor and $(b)$ directly by the definition of \textsf{Step 2} of $\ch'(\cdot)$. Because of $(a)$ and $(b)$ charge gets never moved through 2-real edges in \textsf{Step 1-3} and this is also clear for \textsf{Step 4} what shows $(c)$. Each 1-real triangle $f$ has $\ch_0'(f) = 0$ and contributes no charge in \textsf{Step 1-2}. Thus $\frac{1}{3} - \ch_2'(f) \le \frac{1}{3}$ and so $(d)$ is true. For $(e)$ we consider the 0-real edge $e'$ of the boundary of the 1-real triangle $f$. Since $\overline{e}$ has more than three crossings, $\overline{e'}$ has at most three crossings. Therefore $f$ has a common 1-real edge on its boundary with a 2-real triangle $f'$. 2-real triangles have no 0-real edges and therefore by applying (a) we have $\ch'_1(f) = \ch'_0(f) = \frac{1}{3} > 0$. So $f'$ contributes at least $\frac{1}{6}$ in \textsf{Step 2} to $f$.
        %\item[(c)] A $2$-triangle $f$ can not be a wedge-neighbor of a $0$-triangle or a $1$-triangle. So it can only distribute charge in step 2 and 4. Note that it also can not receive charge in step 1-3. If $f$ is a neighbor of a $1$-triangle, then it contributes $\frac{1}{3}$ in step 2 and therefore it not contributes in step 4.
	\end{nestedproof}
 
    We analyze the final charges $\ch'(f)$ for all faces $f \in F(\Gamma)$. 
    By \cref{le:maximal-min-2-planar-prop}(\ref{3_real_triangle})
    %\cref{pro:min3-prop}$(c)$ 
    we have to consider only $h$-real faces for $h \le 2$ and 3-real triangles. Note that a face contributes in \textsf{Step 1-3} through each edge at most once. Also a face can not get a deficit in \textsf{Step 4}, so it is enough to show $\ch'_3(f) \ge \frac{1}{3} \deg_\Gamma^r(f)$. We use \cref{pro:min3-chargingprop}$(a)$-$(d)$ to receive the following results:
    \begin{itemize}
        \item Each 0-real triangle $f$ has $\ch'_0(f) = -1$, receives $3 \cdot \frac{1}{3}$ in \textsf{Step 1} and does not contribute or receive charge in \textsf{Step 2-3}. So $\ch_3'(f) = 0 = \frac{1}{3} \deg_\Gamma^r(f)$.
        \item Each 1-real triangle $f$ has $\ch'_3(f) =\frac{1}{3} = \frac{1}{3} \deg_\Gamma^r(f)$.
        \item Each 2-real triangle $f$ has $\ch'_0(f) = 1$ and contributes through one or two edges in total at most $\frac{1}{3}$ in \textsf{Step 2}. So $\ch_3'(f) \ge \frac{2}{3} = \frac{1}{3} \deg_\Gamma^r(f)$.
        \item Each 3-real triangle $f$ has $\ch'_0(f) = 2$ and contributes or receives no charge, so $\ch_3'(f) = 2 > 1 = \frac{1}{3} \deg_\Gamma^r(f)$.
        \item Each 0-real quadrilateral $f$ has $\ch'_0(f) = 0$ and contributes or receives no charge, so $\ch_3'(f) = 0 = \frac{1}{3} \deg_\Gamma^r(f)$.
        \item Each 1-real quadrilateral $f$ has $\ch'_0(f) = 1$ and contributes through each 0-real edge at most once and never through a 1-real edge, so it loses in \textsf{Step 1-3} in total at most $\frac{2}{3}$. Therefore $\ch_3'(f) \ge \frac{1}{3} = \frac{1}{3} \deg_\Gamma^r(f)$.
        \item Each 2-real quadrilateral $f$ has $\ch'_0(f) = 2$ and contributes at most $\frac{1}{3}$ through its only 0-real edge and $2 \cdot \frac{1}{6}$ through the 1-real edges, so it loses in \textsf{Step 1-3} in total at most $\frac{2}{3}$. Therefore $\ch_3'(f) \ge \frac{4}{3} > \frac{2}{3} = \frac{1}{3} \deg_\Gamma^r(f)$.
        \item Each $1$-real pentagon $f$ has $\ch'_0(f) = 2$ and contributes through its three 0-real edges at most $3 \cdot \frac{1}{3}$ and through its two 1-real edges at most $2 \cdot \frac{1}{6}$, so $\ch_3'(f) \ge \frac{2}{3} > \frac{1}{3} = \frac{1}{3} \deg_\Gamma^r(f)$.
        \item Each 2-real pentagon $f$ has $\ch'_0(f) = 3$ and contributes through its two 0-real edges at most $2 \cdot \frac{1}{3}$ and through its two 1-real edges at most $2 \cdot \frac{1}{6}$, so $\ch_3'(f) \ge 2 > \frac{2}{3} = \frac{1}{3} \deg_\Gamma^r(f)$.
        \item Each face $f$ with $\deg_\Gamma(f) \ge 6$ has $\ch'_0(f) = \deg_\Gamma(f) - 4 + \deg^r_\Gamma(f)$ and contributes at most $\frac{1}{3}\deg_\Gamma(f)$ charge through its $\deg_\Gamma(f)$ edges. So $\ch_3'(f) \ge \frac{2}{3}\deg_\Gamma(f) - 4 + \deg^r_\Gamma(f) \ge \deg^r_\Gamma(f) \ge \frac{1}{3} \deg^r_\Gamma(f)$.
    \end{itemize}
    It only remains to show $\ch'_4(f) \ge 0$ for all 0-real pentagons $f$. For this we denote for $i \in \{0,1,2,3,4\}$ by $e_i$ the edges of the boundary of $f$ in clockwise order and by $t_i$ (even if not a triangle) the demand-path-neighbor of $f$ at $e_i$. Further we denote by $f_i$ the vertex-neighbors of $f$ at the crossing-vertex of $e_i$ and $e_{(i+1) \mod 5}$.
    %See Figure \ref{fi:min3-definition} for an illustration.
	%\begin{figure}[tb]
    %	\centering
    %	\includegraphics[width=.4\textwidth]{min3-definition}
    %	\caption{A $0$-pentagon $f$ with its edges $e_i$, the wedge-neighbors $t_i$ and the  vertex-neighbors $f_i$. Edges which exist in some (sub)cases, are dotted. Edges with more than three crossings are fat.}\label{fi:min3-definition}
    %\end{figure}
    We consider different cases for the demand-path-neighbors of $f$. Because of rotation- and mirror-symmetry we can shift or negate all indices modulo $5$ without loss of generality and to ease notation we will always use fixed indices. We distinct that more than one demand-path-neighbor of $f$ is a 0-real triangle (case 1 and 2), exactly one 0-real triangle (case 3 and 4) or no 0-real triangle (case 5 and 6). Note that $\ch'_0(f) = 1$, so if $f$ receives $\frac{2}{3}$ or loses not more than $1$ charge then $\ch'_4(f) \ge 0$. This is for example the case if three or less demand-path-neighbors of $f$ are 0-real or 1-real triangles, so we can assume always the opposite. We use this argument repeatedly in the proof.
    
	\paragraph{Case 1. The demand-path-neighbors $t_i$ and $t_{(i+1) \mod 5}$ are 0-real triangles.}
    Fix $i=1$. 
    \begin{figure}[tb]
    	\centering
        \subfigure[]{\includegraphics[width=.3\textwidth]{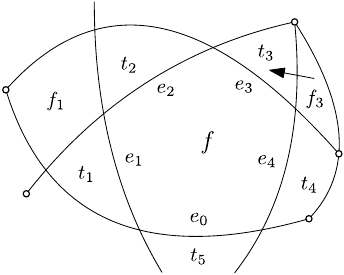}\label{fi:min3-case-1a}}
    	\hfil
    	\subfigure[]{\includegraphics[width=.3\textwidth]{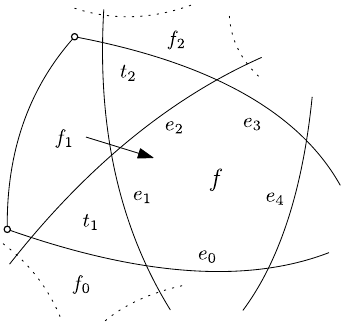}\label{fi:min3-case-1b}}
    	\hfil
    	\subfigure[]{\includegraphics[width=.3\textwidth]{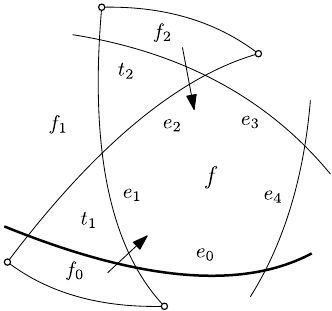}\label{fi:min3-case-1c}}
    	\caption{
    		Case 1: If $\overline{e}_0$ and $\overline{e}_3$ have exactly three crossings, then $f_1$ is a (a) 1-real quadrilateral or (b) a 2-real pentagon. (c) If $\overline{e}_0$ has more than three crossings, then $f$ receives charge from $f_0$ and $f_2$. (\textit{Here and in the following figures, edges that exist in some subcases are dotted and edges with more than three crossings are bold.})}
    	\label{fi:min3-case-1}
    \end{figure}
    Assume that $\overline{e}_0$ and $\overline{e}_3$ have exactly three crossings. Then $f_1$ is a 1-real quadrilateral or a 2-real pentagon.
    \begin{itemize}
    	\item If $f_1$ is a 1-real quadrilateral, then $t_4$ is a 2-real quadrilateral and therefore $\ch'_4(f) \ge 0$, in case one of $t_3$ and $t_5$ is not a triangle. If both are triangles, one of them -- say without loss of generality $t_3$ -- is a 1-real triangle as otherwise $\overline{e}_1$ and $\overline{e}_2$ would have more than three crossings (see \cref{fi:min3-case-1a}). This implies that $t_3$ receives $\frac{1}{3}$ from the 2-real triangle $f_3$ in \textsf{Step 2} and so $\ch'_4(f) \ge 0$.
    	\item If $f_1$ is a 2-real pentagon, then $f_0$ and $f_2$ are not both 0-real pentagons, because then $\overline{e}_1$ and $\overline{e}_2$ would have more than three crossings (see \cref{fi:min3-case-1b}). So $f_1$ distributes its excess of at least $\frac{4}{3}$ in \textsf{Step 4} over at most two faces and therefore $\ch'_4(f) \ge 0$.
    \end{itemize}
	Assume now that without loss of generality $\overline{e}_0$ has four or more crossings (see \cref{fi:min3-case-1c}). Then $\overline{e}_1$ and $\overline{e}_2$ have exactly three crossings and $f_0$ and $f_2$ are 2-real quadrilaterals, which contribute each their excess of at least $\frac{2}{3}$ to at most two faces in \textsf{Step 4}. So $f$ receives $2 \cdot \frac{1}{3}$ and therefore $\ch'_4(f) \ge 0$.
	
	\paragraph{Case 2. The demand-path-neighbors $t_i$ and $t_{(i+2) \mod 5}$ are 0-real triangles.}
    Fix $i=1$. Then $\overline{e}_2$ has four or more crossings and so $\overline{e}_0$ and $\overline{e}_4$ have exactly three crossings. So $f_4$ is a 2-real triangle and contributes in \textsf{Step 2} and \textsf{Step 4} in total $\frac{1}{3}$ to $f$ and its demand-path-neighbors. If $f$ has four or less demand-path-neighbors that are 0-triangles or 1-triangles, then already $\ch'_4(f) \ge 0$.\\
	Assume that $f$ has five such demand-path-neighbors and $t_0, t_2, t_4$ are 1-real triangles (see \cref{fi:min3-case-2}) as otherwise we can refer to case 1. 
 	\begin{figure}[tb]
    	\centering
    	\includegraphics[width=.4\textwidth]{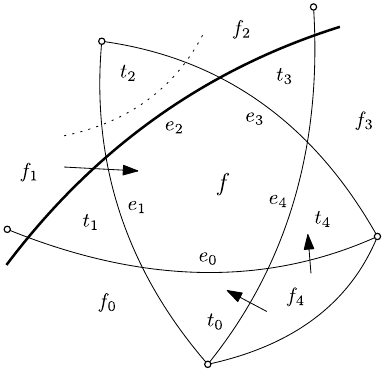}
    	\caption{Case 2: If $t_1$ and $t_3$ are 0-real triangles, then $f_4$ is a 2-real triangle. $f$ receives also charge from $f_1$.}\label{fi:min3-case-2}
    \end{figure}
    If $f_1$ is a 2-real quadrilateral, then it contributes its excess of at least $\frac{2}{3}$ to $f$ in \textsf{Step 4} and so $\ch'_4(f) \ge 0$. If $f_1$ is a 1-real quadrilateral, then it has an excess of $\frac{1}{3}$, because it contributes no charge in \textsf{Step 1} to $f_2$. This excess is contributed in \textsf{Step 4} to only $f$, so $\ch'_4(f) \ge 0$. For $\deg_\Gamma(f_1) \ge 5$ the excess of $f_1$ after \textsf{Step 3} is at least
	\begin{align*}
	\deg_\Gamma(f_1) - 4 + \frac{2}{3} \deg_\Gamma^r(f_1) - \frac{1}{3}\deg_\Gamma(f_1) + 2 \cdot \frac{1}{3} \ge \frac{2}{3} \deg_\Gamma(f_i) - \frac{8}{3},
	\end{align*}
    because $f_1$ contributes through the 1-real edges in total at most $\frac{1}{3}$, it contributes no charge to $f_2$ and $\deg_\Gamma^r(f_1) \ge 1$. This is distributed over at most $\deg_\Gamma(f_1) - 3$ faces as $f_1$ has at most $\deg_\Gamma(f_1) - 1$ vertex-neighbors and does not contribute charge to $f_0$ and $t_2$. For $\deg_\Gamma(f_1) \ge 5$ the equation
	\begin{align*}
	\frac{\frac{2}{3} \deg_\Gamma(f_1) - \frac{8}{3}}{\deg_\Gamma(f_1) - 3} \ge \frac{1}{3}
	\end{align*}
	holds. This implies $\ch'_4(f) \ge 0$.

	\paragraph{Case 3. The demand-path-neighbor $t_i$ is a 0-real triangle and the demand-path-neighbor $t_{j}, j \in \{(i+1) \mod 5, (i-1) \mod 5\}$ is a 1-real triangle with no 0-real quadrilaterals in its demand path.}
    Fix $i = 1$ and $j = 2$. 
	\begin{itemize}
		\item Assume $\overline{e}_0$ has more than three crossings. So $\overline{e}_2$ has exactly three crossings. Then $f_2$ is a 2-real triangle contributing $\frac{1}{3}$ to $f$ and its demand-path-neighbors (see \cref{fi:min3-case-3a}). 
        \begin{figure}[tb]
        	\centering
        	\subfigure[]{\includegraphics[width=.4\textwidth]{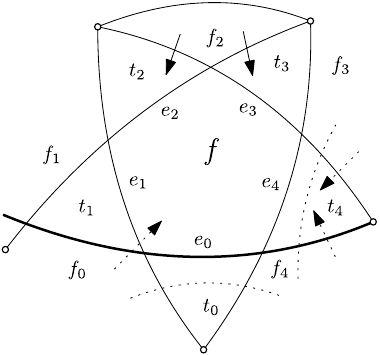}\label{fi:min3-case-3a}}
        	\hfil
        	\subfigure[]{\includegraphics[width=.4\textwidth]{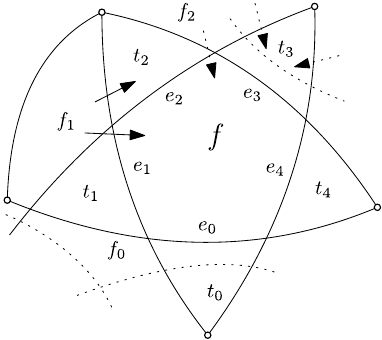}\label{fi:min3-case-3b}}
        	\caption{
        		Case 3: (a) If $\overline{e}_0$ has more than three crossings, then $f_2$ contributes charge to $t_2$ and $t_3$. Also charge is moved to $t_4$ and from $f_0$ to $f$. (b) If $\overline{e}_0$ has exactly three crossings, then $f_1$  contributes charge to $f$ and $t_2$. If this is not $\frac{2}{3}$, then charge is moved to $t_3$ or from $f_2$ to $f$.}
        	\label{fi:min3-case-3}
        \end{figure}
        So we can assume that all demand-path-neighbors of $f$ are 0-real or 1-real triangles (otherwise $\ch'_4(f) \ge 0$) and $t_0, t_2, t_3, t_4$ are 1-real triangles (otherwise we can refer to case 1-2). Note that if the demand path of $t_4$ contains no 0-real quadrilateral, then $f_3$ is a 2-real triangle and therefore $\ch'_4(f) \ge 0$. If it contains one or more 0-real quadrilaterals, then $t_4$ receives at least $\frac{1}{6}$ by \cref{pro:min3-chargingprop}(e). Also $f$ receives at least $\frac{1}{6}$ from $f_0$, as we can see with the following argument: If $f_0$ is a 1-real quadrilateral, then $t_0$ receives $\frac{1}{3}$ charge from a 2-real triangle in \textsf{Step 2}. If $f_0$ is a 2-real quadrilateral, then $f$ receives the excess of $f_0$ in \textsf{Step 4}, which is at least $\frac{2}{3}$. If $f_0$ has more than four vertices, then it distributes not less than $\deg_\Gamma(f_0) - 3 - \frac{\deg_\Gamma(f_0)}{3} = \frac{2}{3} \deg_\Gamma(f_0)-3$ over at most $\deg_\Gamma(f_0) - 3$ vertex-neighbors and so $f$ receives at least $\frac{1}{6}$ charge in \textsf{Step 4}. So in total $\ch'_4(f) \ge 0$.
		\item Assume $\overline{e}_0$ has exactly three crossings. Then $f_1$ is a 2-real quadrilateral, which contributes $\frac{1}{6}$ in \textsf{Step 2} to $t_2$ and has an excess of at least $\frac{2}{3}$ after \textsf{Step 3}, which is distributed over at most two faces in \textsf{Step 4} (see \cref{fi:min3-case-3b}). If there is a demand-path-neighbor of $f$, which is not a 0-real or 1-real triangle, then $\ch'_4(f) \ge 0$. So we assume that $t_0, t_3, t_4$ are 1-real triangles (otherwise we can refer to case 1 or 2).\\
		If now $f_0$ is not a 0-real pentagon, then $f_1$ contributes its excess in \textsf{Step 4} only to $f$ and $\ch'_4(f) \ge 0$ follows. If $f_0$ is a 0-real pentagon, then $\overline{e}_2$ has more than three crossings. So if the demand path of $t_3$ contains one or more 0-real quadrilaterals, then $t_3$ receives at least $\frac{1}{6}$ in by \cref{pro:min3-chargingprop}(e) and $\ch'_4(f) \ge 0$. If the demand path of $t_3$ contains no 0-real quadrilaterals, then $f_2$ is a 2-real triangle that contributes $\frac{1}{6}$ both to $t_2$ and $t_3$ and so $\ch'_4(f) \ge 0$.
	\end{itemize}

	\paragraph{Case 4. The demand-path-neighbor $t_i$ is a 0-real triangle and the demand-path-neighbor $ t_{j}, j \in \{(i+1) \mod 5, (i-1) \mod 5\}$ is a 1-real triangle with one or more 0-real quadrilaterals in its demand path.}
    Fix $i = 1$ and $j = 2$.
	\begin{itemize}
		\item Assume $\overline{e}_0$ has more than three crossings. So $\overline{e}_1$ and $\overline{e}_2$ have exactly three crossings. It follows that $f_0$ is a 2-real quadrilateral with an excess of at least $\frac{2}{3}$ after \textsf{Step 3}, so it contributes at least $\frac{1}{3}$ to $f$ in \textsf{Step 4} (see \cref{fi:min3-case-4a}). 
        \begin{figure}[tb]
        	\centering
        	\subfigure[]{\includegraphics[width=.3\textwidth]{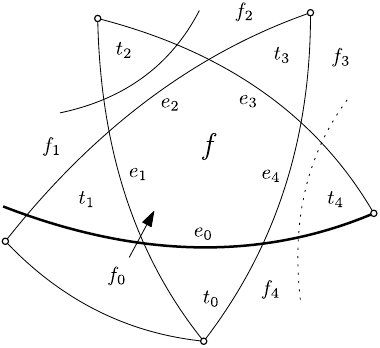}\label{fi:min3-case-4a}}
        	\hfil
        	\subfigure[]{\includegraphics[width=.3\textwidth]{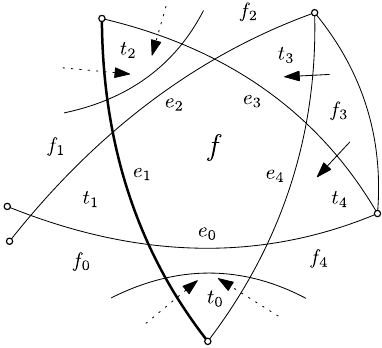}\label{fi:min3-case-4b}}
        	\hfil
        	\subfigure[]{\includegraphics[width=.3\textwidth]{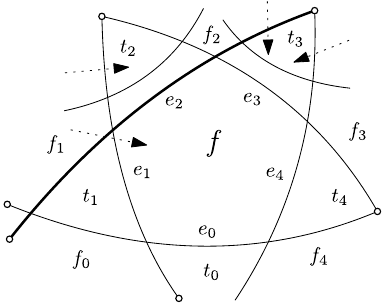}\label{fi:min3-case-4c}}
        	\caption{
        		Case 4: (a) If $\overline{e}_0$ has more than three crossings, then $f_0$ contributes charge to $f$ and we can refer to case 1-3. (b) The situation for the subcase that $\overline{e}_0$ has exactly three crossings and $t_0$ is a 1-real triangle.
        		(c) Situation that $\overline{e}_0$ has exactly three crossings and $t_0$ is not a 1-triangle.}
        	\label{fi:min3-case-4}
        \end{figure}
        So if four or less demand-path-neighbors of $f$ are 0-real or 1-real triangles, then $\ch'_4(f) \ge 0$. Otherwise we assume that all demand-path-neighbors except $t_1$ are 1-real triangles as with another 0-real triangle we can refer to case 1 or 2. But then $t_0$ is a 1-real triangle without a 0-real quadrilateral in its demand path and we can refer to case 3.
		\item Assume $\overline{e}_0$ has exactly three crossings and $t_0$ is a 1-real triangle. If the demand path of $t_0$ has no 0-real quadrilaterals, then we refer to case 3. Otherwise $\overline{e}_1$ has more than three crossings and so by applying \cref{pro:min3-chargingprop}(e) both $t_0$ and $t_2$ receive each $\frac{1}{6}$ charge in \textsf{Step 2} (see \cref{fi:min3-case-4b}). So if not all demand-path-neighbors of $f$ are 0-real or 1-real triangles, then $\ch'_4(f) \ge 0$. So assume now the opposite and that $t_3$ and $t_4$ are 1-triangles (otherwise we can refer to case 1-2). Because $\overline{e}_0$ and $\overline{e}_2$ already have three crossings, $t_3$ and $t_4$ have no 0-real quadrilaterals in their demand paths. So $f_3$ is a 2-real triangle contributing $\frac{1}{6}$ to each $t_3$ and $t_4$. Therefore $\ch'_4(f) \ge 0$.
		\item Assume $\overline{e}_0$ has exactly three crossings and $t_0$ is not a 1-real triangle. Then we can assume that $t_0$ is also not a 0-real triangle (otherwise we can refer to case 1) and that $t_3$ and $t_4$ are 1-real triangles (otherwise we can refer to case 2). If then the demand path of $t_3$ contains no 0-real quadrilateral, it follows that $f_3$ is a 2-real triangle contributing $\frac{1}{6}$ each to $t_3$ and $t_4$ and so $\ch'_4(f) \ge 0$.\\
        If otherwise the demand path of $t_3$ contains a 0-real quadrilateral, by \cref{pro:min3-chargingprop}(e) $t_3$ receives $\frac{1}{6}$ charge in \textsf{Step 2} (see \cref{fi:min3-case-4c}). We now consider different cases for $f_1$. If $f_1$ is a 1-real quadrilateral, then $t_2$ receives $\frac{1}{3}$ from a 2-real triangle in \textsf{Step 2} and so $\ch'_4(f) \ge 0$. If otherwise $\deg_\Gamma(f_1) \ge 5$, then the excess of $f_1$ after \textsf{Step 3} is at least $\deg_\Gamma(f) - 3 - \frac{\deg_\Gamma(f)}{3} = \frac{2}{3}\deg_\Gamma(f) - 3$ and this is distributed in \textsf{Step 4} over at most $\deg_\Gamma(f) - 3$ vertex-neighbors. So $f$ receives at least $\frac{1}{6}$ from $f_1$ and so $\ch'_4(f_0) \ge 0$.
	\end{itemize}

	\paragraph{Case 5. Exactly four demand-path-neighbors of $f$ are 1-real triangles and no demand-path-neighbor is a 0-real triangle.}
    Fix the indices so that $t_i, i \in \{1,2,3,4\}$ are 1-real triangles. If the demand paths of $t_2$ and $t_3$ contain no 0-real quadrilaterals, then $f_2$ is a 2-real triangle contributing $\frac{1}{6}$ each to $t_2$ and $t_3$ and this implies $\ch'_4(f) \ge 0$. So assume without loss of generality that the demand path of $t_2$ contains a 0-real quadrilateral. If the demand path of $t_4$ contains at least one 0-real quadrilateral, then $\overline{e}_3$ has more than three crossings and it follows by \cref{pro:min3-chargingprop}(e) that $t_2$ and $t_4$ both receive $\frac{1}{6}$ in \textsf{Step 2} and therefore $\ch'_4(f) \ge 0$.\\
	If otherwise the demand path of $t_4$ contains no 0-real quadrilaterals, then we can assume that the demand path of $t_3$ does so, because otherwise $f_3$ is a 2-real triangle, what implies $\ch'_4(f) \ge 0$. If now the demand path of $t_1$ contains a 0-real quadrilateral we know by \cref{pro:min3-chargingprop}(e) that $t_1$ and $t_3$ both receive $\frac{1}{6}$ in \textsf{Step 2} and therefore $\ch'_4(f) \ge 0$.\\
	So assume that the demand paths of $t_1$ and $t_4$ contain no 0-real quadrilaterals (see \cref{fi:min3-case-5}). If $\overline{e}_1$ has exactly three crossings, then $f_0$ is 2-real triangle, which contributes $\frac{1}{3}$ to $t_1$ and so $\ch'_4(f) \ge 0$. If $\overline{e}_4$ has exactly three crossings, then with the same argument $f_4$ contributes $\frac{1}{3}$ to $t_4$ and so $\ch'_4(f) \ge 0$. If both $\overline{e}_1$ and $\overline{e}_4$ have more than three crossings, $t_2$ and $t_3$ receive by applying \cref{pro:min3-chargingprop}(e) each $\frac{1}{6}$ charge in \textsf{Step 3} and so $\ch'_4(f) \ge 0$.
	\begin{figure}[tb]
    	\centering
    	\includegraphics[width=.4\textwidth]{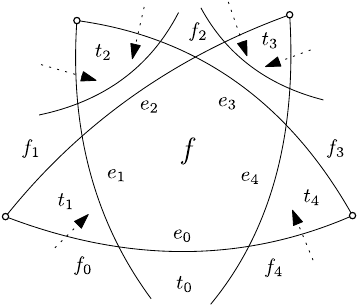}
    	\caption{Case 5: The situation if only $t_2$ and $t_3$ have a 0-real quadrilateral in their demand paths.}\label{fi:min3-case-5}
    \end{figure}
	
	\paragraph{Case 6. All demand-path-neighbors of $f$ are 1-real triangles.}
	If the demand paths of only one demand-path-neighbor $t_i$ or two demand-path-neighbors $t_i, t_{(i+1) \mod 5}$ contain 0-real quadrilaterals, then $f_{(i+2) \mod 5}$ and $f_{(i+3) \mod 5}$ are 2-real triangles and so $\ch'_4(f) \ge 0$. So we assume without loss of generality that the demand paths of $t_1$ and $t_3$ contain 0-real quadrilaterals.\\
	Then $\overline{e}_2$ has more than three crossings and both $t_1$ and $t_3$ receive $\frac{1}{6}$ in \textsf{Step 2} by applying \cref{pro:min3-chargingprop}(e). So if the demand paths of $t_0$ and $t_4$ contain no 0-real quadrilaterals and therefore $f_4$ is a 2-real triangle, it follows $\ch'_4(f) \ge 0$. Thus assume without loss of generality that the demand path of $t_4$ contains a 0-real quadrilateral (see \cref{fi:min3-case-6a}). Then $\overline{e}_0$ has more than three crossings and $t_4$ receives at least $\frac{1}{6}$ in \textsf{Step 2} by applying \cref{pro:min3-chargingprop}(e).\\
    \begin{figure}[tb]
    	\centering
    	\subfigure[]{\includegraphics[width=.4\textwidth]{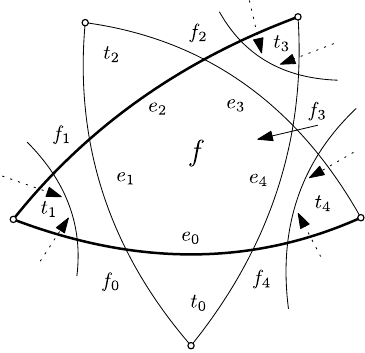}\label{fi:min3-case-6a}}
    	\hfil
    	\subfigure[]{\includegraphics[width=.4\textwidth]{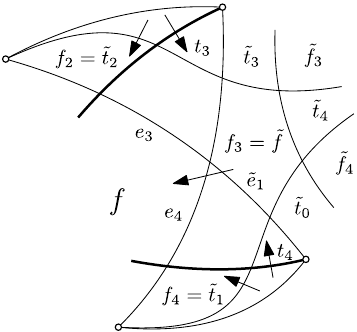}\label{fi:min3-case-6b}}
    	\caption{
    		Case 6: (a) The situation if only $t_1, t_3$ and $t_4$ have a 0-real quadrilateral in their demand paths. (b) Focus on $f_3$ and its renamed neighbors with the further assumption that $f_2$ and $f_4$ are 1-real triangles. Here we show the example that $f_3$ is a 0-real pentagon.}
    	\label{fi:min3-case-6}
    \end{figure}
	Note that the demand paths of $t_3$ and $t_4$ can not contain more than one 0-real quadrilateral and if one of $f_2$ and $f_4$ is a 2-real quadrilateral, it contributes its excess of at least $\frac{2}{3}$ to $f$ in \textsf{Step 4} and so $\ch'_4(f) \ge 0$. So we can assume the opposite. We consider different cases for $f_3$ (see \cref{fi:min3-case-6b} for the example $f_3$ is a 0-real pentagon):
	\begin{itemize}
		\item $f_3$ is a 0-real quadrilateral: Then $f_2$ and $f_4$ are 1-real triangles and therefore $t_3$ and $t_4$ receive both $\frac{1}{6}$ from a 2-real triangle in \textsf{Step 2}. Further they receive also $\frac{1}{6}$ from a 2-real quadrilateral in \textsf{Step 2} and so $\ch'_4(f) \ge 0$.
		\item $f_3$ is a 1-real quadrilateral: Then $t_3$ and $t_4$ receive both $\frac{1}{3}$ charge in \textsf{Step 2} from a 2-real triangle and so $\ch'_4(f) \ge 0$.
		\item $\deg_\Gamma(f_3) \ge 5$: We rename for this purpose the faces in the neighborhood of $\tilde{f}:= f_3$ so that we denote for $i \in \{0, 1, \dots, \deg_\Gamma(\tilde{f}) \}$ by $\tilde{e}_i$ the edges of the boundary of $\tilde{f}$ in clockwise order and we introduce further $\tilde{e}_1$ so that $\overline{\tilde{e}}_1 = \overline{e}_3$, $\tilde{t}_1:=f_4$ and so on. The demand-path-neighbors $\tilde{t}_1$ and $\tilde{t}_2$ are 1-real triangles, which receive both $\frac{1}{6}$ from a 2-real triangle. If further $\tilde{t}_3$ is a 1-real triangle, then it receives (together with $\tilde{t}_4$,  if $\tilde{t}_4$ is a 1-real triangle) $\frac{1}{3}$ charge from a 2-real triangle. The same is true for $\tilde{t}_0$ (together with $\tilde{t}_{\deg_\Gamma(f_3)-1}$). So the excess of $f_3$ is at least
		\begin{align*}
			\deg_\Gamma(f_3) - 4 + \frac{2}{3} \deg_\Gamma^r(f_3) - \frac{1}{3}\deg_\Gamma(f_3) + 2 \cdot \frac{1}{6} + 2 \cdot \frac{1}{3}
			\ge \frac{2}{3} \deg_\Gamma(f_3) - 3
		\end{align*}
		and this is distributed over at at most $\deg_\Gamma(f_3) - 4$ faces. So for $\deg_\Gamma(f_3) \ge 5$ the face $f_3$ contributes at least $\frac{1}{6}$ to each of its vertex-neighbors and this implies with the other charges $\ch'_4(f) \ge 0$.
	\end{itemize}

    \paragraph{Summary} Since we have shown for all 0-real pentagons $f$ that we have $\ch'(f) \ge 0$ in all cases, we know that \textsf{(C1)} holds for all faces w.r.t. $\alpha = \frac{1}{3}$. Together with \cref{eq:modified-charge} and the observation that \textsf{(C2)} is true, \cref{th:density-min3} follows.
\end{proof}

\thheavyminthree*
\begin{proof}
    We apply a similar charging technique as in \cref{th:heavy-min2}, to achieve a bound for the number of heavy edges. Here is a sketch of the proof. From $G$ with a corresponding drawing $\Gamma$, we derive graph $G^*$ with drawing $\Gamma^*$ with the same number of vertices and the same heavy edges crossing the same edges as in $\Gamma$. With the heavy edges being removed, $\Gamma^*$ induces a 2-planar drawing.
    Different than before, 0-real and 1-real faces might occur, and we also cannot guarantee the existence of a triangulation. But fortunately, we can redistribute our total charge of $4n-8$ similar as before, such that each heavy edge receives charge 2 and that implies the claimed upper bound.
    %start segments of each heavy edge receive $\frac 1 4$ charge, while the intermediate segments receive $\frac 1 2$.
    %In total, this gives a minimum of $2 \cdot \frac 1 4 + 3 \cdot \frac 1 2 = 2 $ charge per each heavy edge, which concludes the claimed bound. 
    %Depending on the number of adjacent vertices, we distinguish 0-real triangles, 1-real triangles, 2-real triangles and 3-real triangles. Then we consider the faces which are traversed by heavy edges.
    %We observe that %0-real triangles and 1-real triangles are not traversed at all by heavy edges, therefore they will be ignored in the following. 
    %this triangulation contains only 2-real and 3-real triangles, to which we give as before 1 respectively 2 charge, in total $4n - 8$. 
    %We distribute the charge of those faces to the traversing segments of the heavy edges such that each edge receives at least $\frac{10}{3}$ charge. From this, we can conclude an upper bound of $\frac{6}{5}(n-2)$ heavy edges.

More formally,
let $G$ be any min-3-planar graph with a corresponding drawing $\Gamma$ with the maximum number of heavy edges.
Note that all heavy edges in $\Gamma$ have at least four crossings with light edges.
From $\Gamma$, we derive a graph $G^-$ and drawing $\Gamma^-$ by removing all the heavy edges in $\Gamma$. Further we remove all light edges that have three crossings in $\Gamma^-$ which is 2-planar by construction. Note that the edges that have been originally crossed by the heavy edges will not be removed during this phase. %The new drawing $\Gamma'$ is 2-planar. %Observe that inserting heavy edges would not violate min-1-planarity. 
%All the faces of $\Gamma'$ describe cyclic sequences of real vertices and crossing vertices, and in each sequence we do not have two subsequent crossing vertices, as this would mean two crossings on a light edge.

The next claim describes how to construct $\Gamma^*$ from the drawing $\Gamma^-$. We apply the operation whenever it is possible. 

\begin{claim}
    Let $\Gamma^+$ be the drawing obtained from $\Gamma^-$ by adding the heavy edges. If there exist two light edges $u_1v_1, u_2v_2$ with endpoints $u_1,u_2,v_1,v_2$ having a crossing $x$ and if additionally the two edges have in total $c \ge 1$ crossings in $\Gamma^+$ between $u_1$ and $x$ and between $u_2$ and $x$, we can extend $\Gamma^-$ by the light edge $v_1v_2$ such that it forms a 2-triangle $v_1v_2x$, if it does not exist. Then the achieved graph $G^*$ with the drawing $\Gamma^*$ is still min-2-planar and the 2-real triangle $v_1v_2x$ is crossed by at most one heavy edge if $c=3$ and by no heavy edge if $c\ge 4$.
    
    %can be extended to $\Gamma^*$ by inserting first the heavy edges and then light edges between two vertices $u_1,u_2$, for which two adjacent light edges $e_1, e_2$ have a crossing $x$ and each two
    %such that $\Gamma^*$ achieved is still min-3-planar.
    %If two light edges in $\Gamma'$ have 3 crossings each, and they cross each other such that the crossing $c$ is the last one before the
    %endpoints of the edges, then we can assume that there is an uncrossed edge between the two endpoints, and there is an empty 2-real triangle formed by crossing $c$ and the two endpoints. 
\end{claim}
 \begin{nestedproof}
    Let $c\ge 1$ the number of crossings as defined in the claim. Note that both $u_1v_1$ and $u_2v_2$ have also the crossing $x$. Therefore the two edges might have in total $6-2-c \le 3$ possible crossings at the 2-real triangle $v_1v_2x$. Since each heavy edge that crosses $v_1v_2x$ must cross one of these edges, there can be at most three such heavy edges. Hence $v_1v_2$ is crossed at most three times and $\Gamma^*$ is min-3-planar.
    For $c = 3$ respectively $c=4$ the same argument leads to one respectively zero heavy edges crossing the 2-real triangle $v_1v_2x$.
     %(FIGURE ?)
 \end{nestedproof}
%Diesen eindeutingen Kreuzungsnachbarn nennen wir \tilde{t}

As before, we distinguish $k$-real faces in $\Gamma^*$, where $k$ denotes the number of real vertices in the corresponding face. According to the formula from \cref{eq:initial-charge}, we assign charges to the faces. Our final goal is to distribute $\frac{1}{4}$ charge to each start segment of a heavy edge and $\frac{1}{2}$ charge to each intermediate segment. Furthermore we have to ensure that the charge of each face is non-negative. %From this assignment and the fact, that each heavy edge traverses two start segments and three intermediate segments we can directly conclude an upper bound of $\frac{4(n-2)}{4 \cdot \frac{2}{4}} = 2(n-2)$ heavy edges.

In the following, we discuss different faces:
Note that we will move charges from one face $f$ to another $f'$, if necessary. This will happen at crossings where $f$ and $f'$ do not have adjacent sides but they are opposite at the crossings. Since we ensure that $f$ will always be a 2-real triangle, it has a unique crossing-vertex and therefore it will not lose charge more than once. This general strategy will ensure the correctness of the transfer of the charges. There will be one subcase deviating from this strategy that will be discussed separately. %Here we will move charge from a 2-triangle to a 1-triangle having a different configuration and we will argue there that it still contributes charge once.

\begin{itemize}
    \item 0-real triangles $t$: We start with an initial charge of $-1$. Since each of the three light edges is crossed twice %, their endpoints are directly connected by segments to the three crossings.
    %The start segments involved in the same crossing induce a 2-real triangle $t_1$, since we assume that the endpoints are connected by an (uncrossed) light edge. $t$ has charge 1, and pass $1/2$ of it to the 0-real triangle.
    we know there are three 2-real triangles adjacent to the crossing-vertices of $t$. These triangles are crossed in total by at most three heavy edges. So they can give together $\frac{3}{2}$ charge to $t$ and keep $\frac{3}{2}$ for their three possible intermediate segments of heavy edges.
    %This triangle will get 3 of those $\frac 1 2$ charges, 
    Hence the final charge of $t$ is non-negative. Even if the 0-real triangle is traversed by a heavy edge (note that there can be only one), the corresponding segment can receive the remaining $\frac 1 2$ charge.
    %Note that we left $frac 1 2$ charge at each of the     adjacent 2-real triangles.
    
    \item 1-real triangles $t$: We start with an initial charge of 0. Let $u$ be the single vertex adjacent to the 1-real triangle, and let $e$ be the unique edge at the triangle which is not incident to $u$. Let $c_1$ and $c_2$ be the crossings of the edge $e$ at the boundary of the triangle. Call the endpoints of the edge $e$ $x$ and $y$. The vertices $u$, $c_1$ and $x$ and also the vertices $u$, $c_2$ and $y$ form two 2-real triangles $t_1$ and $t_2$.
    Assume first that $t$ has a heavy edge $e'$ incident to $u$. Then no heavy edges cross $t_1$ and $t_2$, because $e$, $uc_1$ and $uc_2$ cannot have another heavy edge crossing.

    \medbreak
    %deviating subcase
    Here is the critical case for the transfer of the charge: We can charge the start segment of $e'$ by $\frac{1}{4}$ from one of the two 2-triangles that are opposite to $t$ at crossing $c_1$ respectively $c_2$, if not both are crossed twice, as then it has $\frac{1}{2}$ charge left for its one crossing. If both such 2-real triangles are crossed twice, we know that the face $f'$ that is opposite to $c_1$ and $c_2$ at $t_1$ respectively $t_2$ has at least $\deg_{\Gamma^*} \ge 4$ and two real vertices. We will never move charge to this face and therefore $t_1$ and $t_2$ can safely charge the start segment of $e'$ by $\frac{1}{4}$.
    
    Assume now that exactly one heavy edge $e'$ has two crossings with $t$. Then one of the faces opposite at $c_1$ or $c_2$ is a 2-real triangle, that can give $\frac{1}{2}$ charge to $t$. This is true as we know from the above claim that the 2-real triangles exist and if both would have two crossing heavy edges, there would be two crossing edges with more than three crossings which is not allowed.
    
    %We see that this is true with the following argument. Either $e'$ is crossing one of the 1-real edges of $t$ and $e$ and then no heavy edge can cross at $t_1$ or $t_2$ at $e$ and at most one at $\overline{uc_i}$. Or $e'$ is crossing $c_1$ and $c_2$ and if then both $t_1$ and $t_2$ are crossed two times by heavy edges, $e$ or one of $\overline{uc_1}$ and $\overline{uc_2}$ have more than three crossings.
    
    %not both $t_1$ and $t_2$ are crossed two times as this would imply that $e$ has more than three crossings. So one of $t_1, t_2$ can charge the segment of $e'$ at the 1-real triangle by $\frac{1}{2}$, since it is crossed only once by a heavy edge.
    In the last case, we have two heavy edges crossing the 1-real triangle $t$ twice. No matter if one of the heavy edges crosses $e$ or not, we can guarantee the existence of a 2-real triangle $t'$, which is opposite to $t$ at $c_1$ or $c_2$ and which is traversed by no heavy edge. Since this 2-real triangle $t'$ has no heavy edge, we can safely move $1 = 2 \cdot \frac 1 2$ charge to the 1-real triangle, in particular to the two segments of the heavy edges crossing the 1-real triangle $t$.
     \medbreak
    \item 2-real triangles $t$ with crossing $c$: We start with an initial charge of 1, which means that we have enough charge available, if there are at most 2 heavy edges that cross. If three heavy edges cross, then there is a 2-real triangle $t'$ opposite to $t$ at crossing $c$. Note that $t'$ is crossed by at most one heavy edge. Again we can move $\frac{1}{2}$ charge from $t'$ to the third heavy edge that crosses $t$. Note that $t'$ has still enough charge that might be assigned to the single heavy edge of $t'$ if it exists.  
     \medbreak
    \item 3-real triangles $t$: We start with an initial charge of 2, which means that we have enough charge for 4 possible heavy edges. Observe that we can have at most 4 heavy edges which possibly cross $t$. Hence we do not need to transfer any extra charge from somewhere else.
    \item 0-real quadrilaterals $q$: We start with an initial charge of 0, which means that we need to collect either $\frac 1 2$ or $1$ charge depending on the number of possible crossing heavy edges. In both cases (1 or 2 crossing edges), we can argue that there is at least one of the four possible 2-real triangles adjacent to the four crossings of $q$ that is empty and can give its charge to the crossing heavy edges.
    \item 1-real quadrilaterals $q$: We start with an initial charge of 1, so we only need additional charge for the case of 3 crossing heavy edges. Note that more heavy edges can not cross $q$.
    In this critical case, all 2-real triangles adjacent to the three crossings of $q$ are empty and therefore one can give its charge to the third heavy edge of the 1-real quadrilateral.
    \item 0-real pentagons: We start with an initial charge of 1, which means that we have enough charge for 2 possible heavy edges. Observe that we can have at most 2 heavy edges which possibly cross. 
\end{itemize}

For larger faces, it is sufficient to argue only about the number of crossings:
For 2-real quadrilaterals, 1-real pentagons and 0-real hexagons the number of possible crossings on their boundary is not more than a quarter of their charge. This allows us to move $\frac{1}{4}$ to the start segments and $\frac{1}{2}$ to the intermediate segments as they correspond to one respectively two crossings and we avoid therefore negative charge.
This holds also for faces with more real vertices or higher degree, by the following argument. Replacing one crossing by a real vertex increases the charge of a face by 1, while the number of possible crossings increases also by 1 and therefore the required charge increases by only $\frac{1}{4}$. Increasing the degree of a face by 1 increases the charge by 1, while two more crossings are possible. So we can charge the heavy edges that make those additional crossings. 

%Since each heavy edge segment has either one or two crossings, depending whether it is a start segment or an intermediate segment, it receives either charge $\frac 1 4$ or $\frac 1 2$.
This gives in total $2 \cdot \frac 1 4 + 3 \cdot \frac 1 2 = 2$ charge, as we have 2 start segments and at least 3 intermediate segments per heavy edge. From the charge of 2 for each heavy edge, we conclude that we have at most $\frac {4(n-2)} 2 = 2(n-2)$ heavy edges.
%Note that can be considerably improved if there are only small faces in $\Gamma'$, as the start segments of the heavy edges have charge $\frac 1 2$, then.

\medskip

For the lower bound, we refer to the right part of \cref{fi:min-k-heavy-edge-lower-bound}, where we have 6 heavy edges for blocks of 6 vertices (8 vertices from where $u$ and $v$ are removed). Note that this construction can be repeated between the poles $u$ and $v$ similarly as in \cref{fi:8-gons}.
\end{proof}

\subsection{Details for \Cref{se:inclusion-relationships}}\label{se:app-inclusion-relationships}

\thmintwotwoplanar*
\begin{proof}
We first observe that there exist non-optimal min-2-planar graphs that are not 2-planar. For example, $K_{5,5}$ is not 2-planar~\cite{DBLP:journals/jgaa/AngeliniBKS20},  
while \cref{fi:K55} illustrates a min-2-planar drawing of $K_{5,5}$, where black edges have no crossings, orange edges have 1 crossing, blue edges have 2 crossings, green edges have 3 crossings, and the red edge has 4 crossings. 
% %
% \begin{figure}[tb]
%     \centering
%     \includegraphics[width=.45\textwidth]{K55}
% \caption{A min-2-planar drawing of the complete bipartite graph $K_{5,5}$.}\label{fi:K55}
% \end{figure}
% %
In the following, we show how to construct optimal $n$-vertex min-2-planar graphs that are not 2-planar.

% \begin{figure}[tb]
%     \centering
%     \subfigure[]{\includegraphics[page=1,width=.49\textwidth]{truncated-icosahedral-graph}\label{fi:truncated-icosahedral-graph-a}}
%     \hfil
%     \subfigure[]{\includegraphics[page=2,width=.49\textwidth]{truncated-icosahedral-graph}\label{fi:truncated-icosahedral-graph-b}}
%     \caption{
%     (a) A planar drawing $\Gamma$ of the truncated icosahedral graph $G$. (b) A min-2-planar drawing $\Gamma'$ of the graph $G'$, obtained by adding 5 edges to each pentagonal face and 7 edges to each hexagonal face of $\Gamma$.}
%     \label{fi:truncated-icosahedral-graph}
% \end{figure}

Let $G$ be the truncated icosahedral graph and let $\Gamma$ be a planar drawing of~$G$, as depicted in \cref{fi:truncated-icosahedral-graph-a}. This drawing has 12 pentagonal faces, 20 hexagonal faces, 60 vertices and 90 edges.
We enrich $\Gamma$ by adding 5 edges inside each pentagonal face and 7 edges inside each hexagonal face. Denote the obtained graph and the obtained drawing as $G'$ and $\Gamma'$, respectively. $\Gamma'$ is depicted in \cref{fi:truncated-icosahedral-graph-b}, where the edges inside pentagonal faces are colored orange and the edges inside hexagonal faces are colored blue.
More precisely, for each pentagonal face we add an edge between each pair of vertices of the face that are not connected. For each hexagonal face $f$, we add 7 edges as follows; refer to \cref{fi:truncated-icosahedral-graph-b} for an illustration, where the vertices of $f$ are denoted as $u$, $v$, $w$, $x$, $y$, and $z$ . We add an edge between each pair of vertices having distance two on the boundary of $f$. Additionally, we arbitrarily choose two vertices having the maximum distance on the boundary of $f$ ($w$ and $z$ in \cref{fi:truncated-icosahedral-graph-b}) and we add an edge between them, which we call the \emph{diagonal of $f$}. Note that the end-vertices of the diagonal of $f$ have degree 5 in $f$. All the diagonals are dashed in \cref{fi:truncated-icosahedral-graph-b}.

Observe that: $(i)$ $G'$ has $n=60$ vertices and $m=90 + 12 \cdot 5 + 20 \cdot 7 = 290$ edges, thus $m = 5n-10$; $(ii)$ each edge of $G'$ added inside a pentagonal face of $\Gamma$ has two crossings in $\Gamma'$; $(iii)$ for each hexagonal face $f$ of $\Gamma$, the two edges that cross the diagonal (bold in \cref{fi:truncated-icosahedral-graph-b}) have three crossings each, while the other edges added inside $f$ have two crossings each; $(iv)$ no two edges with three crossings cross each other. This implies that $G'$ is optimal min-2-planar and $\Gamma'$ is a min-2-planar drawing of $G'$.

To show that $\Gamma'$ is not 2-planar, we exploit a property on the degree distribution of optimal 2-planar graphs from \cite{DBLP:conf/gd/Forster0R21}, which states that the degree of each vertex of a 2-planar graph is a multiple of three. In what follows, we show that~$G'$ contains vertices whose degree is not a multiple of three.

Each vertex of $G$ belongs to the boundary of two hexagonal faces and one pentagonal face, and it has degree 3. We now show that each vertex of $G'$ has degree 9, or 10, or 11, and that not all of them have degree 9.
Let $u$ be a vertex of $G'$ whose degree is 9 (see, e.g., vertex $u$ in \cref{fi:truncated-icosahedral-graph-b}). Since $u$ belongs to the boundary of two hexagonal faces and one pentagonal face in $G$, and it has degree 9, it cannot be incident to any diagonals (otherwise it would have degree larger than 9). This implies that at least one of the vertices that are adjacent to $u$ in $G$, call it $v$, is the end-vertex of at least one diagonal in $G'$. Two cases are possible: $(a)$ $v$ is the end-vertex of one diagonal; $(b)$ $v$ is the end-vertex of two diagonals. In case $(a)$, $v$ has degree 10 (see, e.g., vertex $v$ in \cref{fi:truncated-icosahedral-graph-b}); in case $(b)$, $v$ has degree 11 (see, e.g., vertex $z$ in \cref{fi:truncated-icosahedral-graph-b}).
\end{proof}

\thminthreethreeplanar*
\begin{proof}
%
% \begin{figure}[tb]
%     \centering
%     \subfigure[]{\includegraphics[page=2,width=.49\textwidth]{eight-gons}\label{fi:8-gons-a}}
%     \hfil
%     \subfigure[]{\includegraphics[page=1,width=.49\textwidth]{eight-gons}\label{fi:8-gons-b}}
%     \caption{
%     Illustration of the construction of \Cref{le:min3-3planar}. If in the graph of figure (a) we replace each shaded chain with a copy of the graph of figure (b), we get a min-3-planar graph that is not 3-planar. The bold edges are heavy edges.}
%     \label{fi:8-gons}
% \end{figure}
%
First, consider a planar graph $G$, and a corresponding drawing $\Gamma$, consisting of $h$ parallel chains ($h \geq 1$), each with $8$ vertices, sharing the two end-vertices $u$ and $v$, and interleaved by $h$ copies of edge $uv$; refer to \cref{fi:8-gons-a}. Then, construct a new graph $G'$, and a corresponding drawing $\Gamma'$, obtained from~$G$, and from $\Gamma$, by replacing each parallel chain with a copy of the graph $G''$ depicted in \cref{fi:8-gons-b}. In the drawing $\Gamma'$, each copy of $G''$ has the same edge crossings as the drawing illustrated in \cref{fi:8-gons-b}.
Graph $G''$ has $8$ vertices and $33$ edges, and it is min-3-planar. Indeed, only four edges in the drawing of $G''$ shown in the figure have more than three crossings and they do not cross each other (see the bold edges in \cref{fi:8-gons-b}). It follows that $G'$ is min-3-planar; also it has $n=6h+2$ vertices and $m=33h+h=34h$ edges. Since $h=\frac{n-2}{6}$, we have $m=\frac{17}{3}n-\frac{34}{3}=5.\overline{6}n-11.\overline{3}$. Therefore, $m>5.5n-11$ for every $n > \frac{1}{2}$. Since a 3-planar graph has at most $5.5n-11$ edges~\cite{DBLP:journals/dcg/PachRTT06}, $G'$ is not 3-planar.
\end{proof}

\subsection{Details for \Cref{se:open}}\label{se:app-conclusions}

% \leminonegapplanar*
% \begin{proof}
% It suffices to prove the statement for $k=1$.
% Every min-1-planar drawing can be made  1-gap-planar by introducing a gap on every red edge crossing a green edge (see the coloring rule in \cref{sse:edge-density-min-1}). Indeed, by definition, each red edge receives only one gap.
% Finally, if a min-1-planar drawing is not 3-quasi planar, then it has three mutually crossing edges, which implies the existence of two edges that cross at least two times each.
% Also, since optimal 1-gap-planar graphs have $5n-10$ edges \cite{DBLP:journals/tcs/BaeBCEE0HKMRT18} and there exist 3-quasiplanar graphs with $6.5n-O(1)$ edges \cite{ackerman.tardos:on}, the inclusion relationships are proper.
% \end{proof}

\leminkgapplanar*
\begin{proof}
Let $\Gamma$ be a min-$k$-planar drawing, for any given $k \geq 1$.
Each crossing in $\Gamma$ involves at least one light edge, hence the set of light edges covers all crossings in $\Gamma$. Consider any light edge $e$ and assign each crossing of $e$ to $e$. Then, consider a second light edge $e'$ and assign all unassigned crossings of $e'$ to $e'$. Iterate this procedure until all crossings have been assigned to some light edge. Since each light edge has at most $k$ crossings, no more than $k$ crossings are assigned to a single edge. Hence $\Gamma$ is $k$-gap-planar.

We now prove that $\Gamma$ is also $(k+2)$-quasi planar. Suppose by contradiction that this is not the case. This means that $\Gamma$ contains $k+2$ mutually crossing edges. Since no two heavy edges cross, at least $k+1 \geq 2$ of these edges are light edges. But each of them cross $k+1$ times, a contradiction.
\end{proof}

\leminkfanplanar*
\begin{proof}
The existence of min-$k$-planar graphs that are not fan-planar is an immediate consequence of the fact that there exist 2-planar graphs that are not fan-planar \cite{DBLP:journals/tcs/BinucciGDMPST15}. 
%
%pippo%
To show the existence of fan-planar graphs that are not min-$k$-planar, consider the graph $K_{1,3,h}$, for any $h \geq 1$, and let $n=4+h$ be its number of vertices. It is easy to see that this graph is fan-planar (see, e.g.,~\cite{DBLP:journals/tcs/BinucciGDMPST15}). Also, it is known that any drawing of $K_{1,3,h}$ has $\Omega(h^2)=\Omega(n^2)$ crossings \cite{DBLP:journals/jgt/Asano86}.
On the other hand, by \cref{th:density-mink-general}, any min-$k$-planar drawing with $n$ vertices has at most $c \sqrt{k} \cdot n$ edges (for a constant $c$) and therefore, by \cref{pr:crossings-min-k}, it has at most $c k^{1.5}n$ crossings. Hence, $K_{1,3,h}$ is not min-$k$-planar for sufficiently large values of $n$.  
\end{proof}

%fin qui

\end{document}